\documentclass[12pt]{article} %
\usepackage{amsmath,amssymb,amsthm}
\numberwithin{equation}{section} %
\theoremstyle{plain}
   \newtheorem{thm}{\hspace{\parindent}{\sc Theorem}}[section] %
   \newtheorem{pro}[thm]{\hspace{\parindent}Proposition}
   \newtheorem{cor}[thm]{\hspace{\parindent}Corollary}
   \newtheorem{lem}[thm]{\hspace{\parindent}Lemma}
\theoremstyle{remark} %
   \newtheorem{rem}{\hspace{\parindent}Remark}[section] %
\newtheorem{exmp}{\hspace{\parindent}Example}[section]%
\newcommand{\bA}{\mathbf{A}}
\newcommand{\bB}{\mathbf{B}}
\newcommand{\bs}{\widehat{\mathbf{s}}}
\newcommand{\bq}{\mathbf{q}}
\newcommand{\bR}{\mathbb{R}}
\newcommand{\bx}{\mathbf{x}}

\newcommand{\by}{\mathbf{y}}
\newcommand{\Cts}{\mathcal{C}}

\newcommand{\Cspace}{C^{\infty}}

\newcommand{\domain}{[0,T]\times \bR^d}
\newcommand{\dA}{\mathfrak{A}}
\newcommand{\dH}{\mathfrak{H}}
\newcommand{\dR}{\mathfrak{R}}

\newcommand{\kdelta}{K_{\Delta}}
\newcommand{\limepsilon}{\lim_{\epsilon\rightarrow 0+}}
\newcommand{\qdelta}{q_{\Delta}}
\newcommand{\qts}{q^{t,s}_{x,y}}
\newcommand{\sqrho}{\sqrt{\rho}}
\newcommand{\Sspace}{{\cal S}}
\newcommand{\ts}{t,s}
\newcommand{\upl}{^{(l)}}
\def\dbar{{\mathchar'26\mkern-12mud}}
\begin{document}
\title{On the Feynman path integral for the magnetic Schr\"odinger equation with a polynomially growing electromagnetic potential}
\author{Wataru Ichinose\thanks{This work was supported by JSPS KAKENHI Grant Number JP18K03361.}} %
\date{}
\maketitle %
\begin{quote}
{\small Department of Mathematics, Shinshu University,
Matsumoto 390-8621, Japan. \\
 E-mail: ichinose@math.shinshu-u.ac.jp}%
\end{quote}\par
	\begin{abstract}
	The Feynman path integrals for the magnetic Schr\"odinger equations are defined mathematically, in particular,  with polynomially growing potentials in the spatial direction. For example, we can handle  electromagnetic potentials $(V,A_{1},A_{2},\dots,A_{d})$ such that $V(t,x) = |x|^{2(l+1)} + $\\`` a polynomial of degree $(2l + 1)$ in $x$ " ($l = 0,1,2,...$) and  $A_{j}(t,x)$ are polynomials  of degree $l$ in $x$.  The Feynman path integrals are defined as $L^2$-valued continuous functions with respect to the time variable.  	
\end{abstract}

\section{Introduction}%
Let $T > 0$ be an arbitrary constant, $0 \leq t \leq T$ and 
$x = (x_1,\dotsc,x_d)\in \bR^d$.  Let $E(t,x) = (E_1,
\dots, E_d) \in \bR^d$ and $B(t,x) = (B_{jk}(t,x))_{1\leq j < k \leq d}
\in \bR^{d(d-1)/2}$ denote the 
   electric strength  and the magnetic strength tensor, respectively and $(V(t,x),A(t,x)) = (V,A_1,\dots,A_d)  \in \bR^{d+1}$ an electromagnetic potential, i.e.
\begin{align} \label{1.1}
     & E = -\frac{\partial A}{\partial t} - \frac{\partial V}{\partial 
x},\notag \\
          & B_{jk} =  \frac{\partial A_k}{\partial x_j}  -\frac{\partial 
A_j}{\partial x_k}
\quad (1 \leq j <  k \leq d),
\end{align}
where  $\partial V/\partial x = (\partial V/\partial x_1,\dots,\partial V/\partial x_d)$. Then the Lagrangian function is given by
\begin{equation} \label{1.2}
      {\cal L}(t,x,\dot{x})=  \frac{m}{2}|\dot{x}|^2 + e\dot{x}\cdot A(t,x) - eV(t,x),
\   \dot{x}\in \bR^d
\end{equation}
with mass $m > 0$ and charge $e \in \bR$.  Then the corresponding Schr\"odinger equation is given by
\begin{align} \label{1.3}
& i\hbar \frac{\partial u}{\partial t}(t)  = H(t)u(t)\notag\\ 
& := \left[ \frac{1}{2m}\sum_{j=1}^d
      \left(\frac{\hbar}{i}\frac{\partial}{\partial x_j} - eA_j(t,x)\right)^2 + eV(t,x)\right]u(t),
\end{align}
where $\hbar$ is the Planck constant.
Throughout this paper we always consider solutions to the Schr\"odinger equations in the sense of distribution.
Hereafter we suppose $\hbar = 1$ and $e = 1$ for  simplicity. 
\par
 Let $L^2 = L^2(\bR^d)$ denote the space of all square integrable functions on
$\bR^d$ with inner
product $(f,g) := \int f(x)g(x)^*dx$ and  norm $\Vert f\Vert$, where $g(x)^*$ denotes the complex conjugate of $g(x)$. 
Let $S(t,s; q)$ be the classical action 
\begin{equation} \label{1.4}
   S(t,s;q) = \int_s^t {\cal L}(\theta,q(\theta),\dot{q}(\theta))d\theta
\end{equation}
for a path $q(\theta) \in \bR^d\ (s \leq \theta \leq t)$, where $\dot{q}(\theta) = dq(\theta)/d\theta$.  Our aim in the present paper is to prove that for any $f \in L^2$ we can determine the Feynman path integral
\begin{equation} \label{1.5}
 K(t,0)f = \int e^{iS(t,0;q)}f(q(0)){\cal D}q
\end{equation}
in $L^2(\bR^d)$ for the system \eqref{1.2} with a potential $(V,A)$ growing polynomially in the  spatial direction $x$.  As shown in Example 3.1, a typical example of potentials that we can handle is 
\begin{equation} \label{1.6}
V(t,x) = |x|^{2(l+1)} + \sum_{|\alpha| \leq 2l + 1}a_{\alpha}(t)x^{\alpha},
\end{equation}
\begin{equation} \label{1.7}
A_j(t,x) =  \sum_{|\alpha| \leq l}b_{j\alpha}(t)x^{\alpha}\ (j = 1,2,\dots,d)
\end{equation}
with an integer $l \geq 0$ and functions $a_{\alpha}(t) \in \bR, b_{j\alpha}(t) \in \bR$ in $C^1([0,T])$, i.e. continuously differentiable functions, where $|x|^2 = \sum_{j=1}^d x_j^2$ and for a multi-index $\alpha = (\alpha_1,\dots,\alpha_d)$  we  write $|\alpha| = 
\sum_{j=1}^d
\alpha_j$, $x^{\alpha} =  x_1^{\alpha_1}
\cdots  x_d^{\alpha_d}, \partial_{x_j} = \partial /\partial
x_j$ and  $\partial_x^{\alpha} = \partial_{x_1}^{\alpha_1}
\cdots \partial_{x_d}^{\alpha_d}$.
\par
  In the present paper  the Feynman path integral \eqref{1.5} is defined by the time-slicing method in terms of piecewise free moving paths or piecewise straight lines.  The time-slicing approach in terms of piecewise free moving paths to path integrals is actually the classical approach in the physics literature (cf.  p.32 in \cite{Feynman-Hibbs}  and p.278 in \cite{Peskin-Schroeder}).
  \par
  The Feynman path integral for the system \eqref{1.2} with potentials $A = 0$ and $V$ satisfying $|V(t,x)| \leq C(1 + |x|^2)$ has been studied mathematically by many authors for a long time since Feynman had published his famous paper \cite{Feynman} in 1948 (cf. \S 10 in \cite{Albeverio et all} and \cite{Asada-Fujiwara}).  See \cite{Fujiwara 2017}, \cite{Ichinose 1997}, \cite{Ichinose 1999}, \cite{Ichinose 2007}, \cite{Nicola} and their references for the recent study.  We note that in \cite{Ichinose 1997}, \cite{Ichinose 1999} and \cite{Ichinose 2007} we studied the Feynman path integrals defined by the time-slicing method in terms of piecewise free moving paths.
\par 
On the other hand, if $|V(t,x)| \geq C(1 + |x|^2)^{1+\delta}$ holds with positive constants $C$ and $\delta$, it may be not simple to construct the Feynman path integral for \eqref{1.2}  mathematically as stated in \S 10.2  of \cite{Albeverio et all} and  in \S 3.5 of \cite{Mazzucchi}.  In fact, there seems to be only a few papers on it, which will be referred below, as far as the author knows.  In addition, we note that as well known, the uniqueness of solutions to \eqref{1.3} with $u(0) = f$ in the space $C^0_t([0,T];L^2)$  doesn't hold in general if $V(t,x)$ satisfies $V(t,x) \leq -C(1 + |x|^2)^{1+\delta}$ with positive constants $C$ and $\delta$ (cf. pp. 157-159 in \cite{Berezin-Shubin}, Theorem VIII.7 in \cite{Reed-Simon I}, Theorems X.2 and X.3 in \cite{Reed-Simon II}), where $C^0_t([0,T];L^2)$ denotes the space of all $L^2$-valued continuous functions in $t \in [0,T]$.
\par
Nelson in \cite{Nelson} has constructed the Feynman path integral \eqref{1.5} for \eqref{1.2} in $L^2$ for $f \in L^2$ with $A = 0$ and a  continuous function $V(x)$ outside a set of capacity $0$ in $\bR^d$, independent of $t \in [0,T]$,  by using the Trotter product formula. It is to be noted that in \cite{Nelson} the classical action $S(t,0;q)$ is replaced with a certain approximation.  See (9) on p. 333 of \cite{Nelson}. 
\par
  Daubechies and Klauder in \cite{D-K} have showed the following.  Take $(V, A) = (V(x), A(x))$, independent of $t$, satisfying $|V(x)| \leq C(1 + |x|^2)^{M}$ and  $|A(x)| \leq C(1 + |x|^2)^{M}$ with constants $C \geq 0$ and $M \geq 0$.  
  Let  $H$  be the operator defined by \eqref{1.3} with a core consisting of  finite linear spans generated by eigenvalues of $(-\Delta + |x|^2)/2$, where $\Delta = \sum_{j=1}^d\partial_{x_j}^2$.  Let $F_0(x)$ be the ground state of  $(-\Delta + |x|^2)/2$ and define the canonical coherent states $|p,q> = e^{ip\cdot x}F_0(x - q)$ for all $(q,p) \in \bR^{2d}$, where $p\cdot x = \sum_{j=1}^dp_jx_j$.  Let $H'$ be a maximal extension of $H$ on $L^2$ and denote the deficiency indices of $H'$ by $n_+(H')$ and $n_-(H')$.  
 Daubechies and Klauder have constructed the phase space Feynman path integral in the form of weak topology of $L^2$, i.e. giving $(|p'',q''> , e^{-itH'}|p',q'> )$ if $n_+(H') = 0$ and $(|p'',q''> , e^{-itH'^{\dag}}|p',q'> )$ if $n_-(H') = 0$  in terms of
 the Winer measure pinned at $(p',q')$ at $t = 0$ and at $(p'',q'')$ at $t$, where $H'^{\dag}$ denotes the adjoint operator of $H'$. 
 \par
 Albeverio and Mazzucchi in \cite{A-M 2005} and \cite{A-M 2009} have studied the Feynman path integrals for the systems \eqref{1.2} with $A = 0$,  $V(x) = |\Omega x|^2/2+ \lambda C(x,x,x,x) \ (\lambda \in \bR)$ and with $A = 0$,  a positive homogeneous polynomial $V(x)$ of $2M$-order $(M = 1,2,\dots)$, respectively in terms of infinite dimensional oscillatory integrals and the Wiener measure, where $\Omega$ is a $d\times d$ regular matrix and $C(x,y,w,z)$ is  a completely symmetric positive fourth-order covariant tensor on $\bR^d$.  It is noted that all Feynman path integrals in \cite{A-M 2005, A-M 2009} are defined in the form of weak topology of $L^2$.  See \S 10.2 in \cite{Albeverio et all} and \S 3.5 in \cite{Mazzucchi} for topics relating to \cite{A-M 2005, A-M 2009}.
 \par
 The present paper is having four points to be emphasized: (1) In our system \eqref{1.2} there exists a magnetic field $B(t,x)$. (2) Our magnetic field $B(t,x)$ and electric field $E(t,x)$ can vary on time $t$. (3)  Our Feynman path integral can be defined as an $L^2$- valued function on $[0,T]$ as in \cite{Nelson}, not in the form of weak topology of $L^2$ as in \cite{A-M 2005, A-M 2009, D-K}. (4) Our method of constructing the Feynman path integral can not be applied to systems with potentials satisfying $V(t,x) \leq -C(1 + |x|^2)^{1+\delta} \ ( C > 0, \delta >0)$, though  in \cite{A-M 2005, D-K, Nelson} the Feynman path integrals for such systems are constructed.  
  \par
 In the present paper the Feynman path integrals will be constructed not only for the one-particle systems \eqref{1.2}, but also the multi-particle systems with spin.  In addition, we will construct the Feynman path integrals for bosons and fermions, i.e. quantum systems consisting of many identical particles with spin.
 \par
 We will prove the results in the present paper, following the proofs in \cite{Ichinose 1997, Ichinose 1999, Ichinose 2007, Ichinose 2017}.  That is, we introduce the fundamental operator $\mathcal{C}(t,s)$ in \S 5, and prove its stability and consistency.  Combining these results and the existence theorem proved in \cite{Ichinose 2018} to the Schr\"odinger equations \eqref{1.3}  in both of $L^2$ and the Schwartz space $\Sspace(\bR^d)$ of all rapidly decreasing functions on $\bR^d$, we can prove our main results.  In particular,  in the present paper we will use the delicate result below concerning the $L^2$-boundedness of pseudo-differential operators, which is stated as Theorem 13.13 on p. 322 in \cite{Zworski}. 
 \par
 \vspace{0.5cm}
 T{\sc heorem} 1.A.  {\it Suppose $p(x,\xi,x') \in S^0(\bR^{3d})$, i.e. 
 \begin{equation} \label{1.8}
\sup_{x,\xi,x'}|\partial_{\xi}^{\alpha}\partial_{x}^{\beta}\partial_{x'}^{\gamma}p(x,\xi,x')| \leq C_{\alpha,\beta,\gamma} < \infty
 \end{equation}
for all $\alpha, \beta$ and $\gamma$.  Let $P(X,\hbar D_x,X')$ be the pseudo-differential operator defined by 
 \begin{equation*} 
\int e^{ix\cdot \xi}\ \dbar\xi \int e^{-ix'\cdot \xi}p(x,\hbar\xi,x')f(x')dx',\quad \dbar\xi = (2\pi)^{-d}d\xi
 \end{equation*}
for $f \in \Sspace(\bR^d)$.  Then we have
 \begin{equation} \label{1.9}
\Vert P(X,\hbar D_x,X')\Vert_{L^2\to L^2} = \sup_{x,\xi,x'}|p(x,\xi,x')| + O(\hbar),
 \end{equation}
where $\Vert P\Vert_{L^2\to L^2}$ denotes the operator norm from $L^2$ into $L^2$.}
\vspace{0.5cm}
\par
The plan of the present paper is as follows.  In \S 2 our main results are stated.  In \S 3 we will state examples to which our results can be applied.  In \S 4 we will construct the Feynman path integrals for bosons and fermions.  In \S 5 and \S 6 the stability and the consistency of $\Cts(t,s)$ will be proved, respectively.  In \S 7 Theorems 2.1 - 2.2 and in \S 8
Theorems 2.3 - 2.4 will be proved.
\section{Main theorems}
Let $t$ in $[0,T]$.  For an arbitrary integer $\nu \geq 1$ we take $\tau_j \in [0,T]\ (j = 1,2,\dots,\nu-1)$ satisfying $0 = \tau_0 < \tau_1 < \dots <\tau_{\nu-1} < \tau_{\nu}= t$,  set $\Delta := \{\tau_j\}_{j=1}^{\nu-1}$ and write $|\Delta|:= \max\{\tau_{j+1}- \tau_j; j = 0,1,\dots,\nu-1\}$. Let $x \in \bR^d$ be fixed.
  We take arbitrary points
$x^{(j)} \in \bR^d\ (j = 0,1,\dotsc,\nu-1)$ and determine the piecewise free moving path or the piecewise straight line $\qdelta(\theta;x^{(0)},\dotsc,x^{(\nu-1)},x) \in \bR^d \ (0 \leq \theta \leq t)$ by joining  $x^{(j)}$ at $\tau_j\ (j = 0,1, \dotsc,\nu, 
x^{(\nu)} = x)$  in order.  Let $\mathcal{L}(t,x,\dot{x})$ be the Lagrangian function defined by \eqref{1.2} and $S(t, s; q)$ the classical action defined by \eqref{1.4}.  Take $\chi \in \Cspace_0(\bR^d)$, i.e. an infinitely differentiable function on $\bR^d$ with compact support,  such that $\chi(0) = 1$ and determine the approximation  of the Feynman path integral \eqref{1.5} for $f \in \Cspace_0(\bR^d)$ by 
 \begin{align} \label{2.1}
\kdelta(t,0)f =  \limepsilon & \prod_{j=0}^{\nu-1}\sqrt{\frac{m}{2\pi i(\tau_{j+1} - 
\tau_{j})}}^{\ d}
        \int\cdots\int_{\bR^d} e^{iS(t,0;\qdelta)} \notag \\
       & \times f(x^{(0)}) \prod_{j=1}^{\nu-1}\chi(\epsilon x^{(j)})
          dx^{(0)}dx^{(1)}\cdots dx^{(\nu-1)}.
 \end{align}
 From now on we always suppose that $\chi$ is a real-valued function belonging to  $\Cspace_0(\bR^d)$ such that $\chi(0) = 1$.
The RHS is an oscillatory integral and will be denoted by
 \begin{align*} 
  & \prod_{j=0}^{\nu-1}\sqrt{\frac{m}{2\pi i(\tau_{j+1} - 
\tau_{j})}}^{\ d}
      \text{Os} -  \int\cdots\int_{\bR^d} e^{iS(t,0;\qdelta)}  f(x^{(0)}) 
          dx^{(0)}dx^{(1)}\cdots dx^{(\nu-1)}
 \end{align*}
 (cf. p. 45 of \cite{Kumano-go}).
  \par
   In the present paper we often use symbols $C, C_{\alpha}, C_{\alpha,\beta}$, $C_a$  and $\delta$ to write down constants, though these values are different in general. 
   \vspace{0.5cm}
   \par
   {\bf Assumption 2.1. } We assume that $\partial_x^{\alpha}\partial_t^kV(t,x)$ and $\partial_x^{\alpha}\partial_t^kA_j(t,x)\ (j = 1,2,\dots,d)$ are continuous in $\domain$  for all $\alpha$ and $k = 0,1$. Moreover, we assume the existence of constants $M_* \geq 0, C_0 > 0, C_1 \geq 0, C_2 \geq 0$ with
\begin{equation} \label{2.2}
C_0<x>^{2(M_*+1)} - C_1 \leq V(t,x) \leq C_2<x>^{2(M_*+1)}
\end{equation}
in  $\domain$, where $<x> = \sqrt{1 + |x|^2}$.  We also assume
\begin{equation} \label{2.3}
|\partial_x^{\alpha}V(t,x)| \leq C_{\alpha}<x>^{2(M_*+1)}, |\alpha| \geq 1,
\end{equation}
\begin{equation} \label{2.4}
 |\partial_x^{\alpha}\partial_tV(t,x)| \leq C_{\alpha}<x>^{2(M_*+1)}\end{equation}
for all $\alpha$,
\begin{equation} \label{2.5}
 |A_j(t,x)| \leq C<x>^{M_*+1-\delta}
 \end{equation}
with a constant $\delta > 0$ and 
\begin{equation} \label{2.6}
|\partial_x^{\alpha}\partial_tA_j(t,x)| \leq C_{\alpha}<x>^{M_*+1}
 \end{equation}
 for all $\alpha$.
%
%
%
\vspace{0.3cm}\par
 {\bf Assumption 2.2. } Let $M_*$ be the constant in Assumption 2.1.  We assume
\begin{equation} \label{2.7}
   C_*|x|^{2M_*} - C_1 \leq -\frac{1}{2m}\left(\frac{\partial E}{\partial x}(t,x) + \frac{{}^t\partial E}{\partial x}(t,x)\right)
\end{equation}
with constants $C_* > 0$ and $C_1 \geq 0$, and 
 \begin{equation} \label{2.8}
 |\partial_x^{\alpha}E_j(t,x)| \leq C_{\alpha}<x>^{2M_*},\  |\alpha| \geq 1,
 \end{equation}
 where $\partial E/\partial x = (\partial E_i/\partial x_j;i \downarrow j \rightarrow 1,2,\dots,d)$ is a $d\times d$ matrix and ${}^t\partial E/\partial x$ its transposed matrix. We assume
\begin{equation} \label{2.9}
|\partial_x^{\alpha}A_{j}(t,x)| \leq C_{\alpha}<x>^{M_*}, \ |\alpha| \geq 1.
\end{equation}
In addition, we assume either
\begin{equation} \label{2.10}
|\partial_x^{\alpha}B_{jk}(t,x)| \leq C_{\alpha}<x>^{-(1+\delta_{\alpha})}, |\alpha| \geq 1
\end{equation}
with constants $\delta_{\alpha} > 0$ and 
\begin{equation} \label{2.11}
 |\partial_x^{\alpha}\partial_tB_{jk}(t,x)| \leq C_{\alpha}<x>^{M_*}
 \end{equation}
for all $\alpha$, or
\begin{equation} \label{2.12}
  |\partial_x^{\alpha}\partial_tB_{jk}(t,x)| \leq C_{\alpha}<x>^{-(1+ \delta_{\alpha})},\ |\alpha| \geq 1
 \end{equation}
with  constants $\delta_{\alpha} > 0$ if $0 \leq M_{*} < 1$ and 
\begin{equation} \label{2.13}
|\partial_x^{\alpha}\partial_tB_{jk}(t,x)| \leq C_{\alpha}<x>^{M_*-1}, \ |\alpha| \geq 1
 \end{equation}
 if $M_{*} \geq 1$. 
 \vspace{0.5cm}
 \par
\begin{thm}
Suppose that Assumptions 2.1 and 2.2 are satisfied.  Then there exist constants $\rho^{*} > 0$ and $K \geq 0$ such that  the following statements hold for all $\Delta$ satisfying $|\Delta| \leq \rho^{*}$ and all $t \in [0,T]$: 
\par
\noindent (1) $K_{\Delta}(t,0)f$ defined on $ f \in \Cspace_{0}(\bR^{d})$  by \eqref{2.1} is determined independently of the choice of $\chi$ and $K_{\Delta}(t,0)f$ can be uniquely extended to a bounded operator on $L^{2}$ with 
\begin{equation} \label{2.14}
\Vert \kdelta(t,0)f\Vert \leq e^{Kt}\Vert f\Vert
 \end{equation}
 for all $f \in L^{2}$. 
 \par
\noindent (2) For all $f \in L^{2}$, as $|\Delta| \to 0$, $\kdelta(t,0)f$ converges in $L^{2}$ uniformly in $t \in [0,T]$ to an element $K(t,0)f \in L^{2}$, which we call the the Feynman path integral of $f$.  
\par
\noindent (3) For all $f \in L^{2}$, $K(t,0)f$  belongs to   $C_{t}^{0}([0,T];L^{2})$. In addition, $K(t,0)f$ is the unique solution  in $C^{0}_{t}([0,T];L^{2})$ to \eqref{1.3} with $u(0) = f$. 
\par
\noindent (4) Let $\psi(t,x)$ be a real-valued function such that $\partial_{x_j}\partial_{x_k}\psi(t,x)$ and $\partial_{t}\partial_{x_j}\psi(t,x)$ $(j,k = 1,2,\dotsc,d)$ are continuous in $[0,T] \times \bR^d$ and consider the gauge transformation 
\begin{equation}  \label{2.15}
    V' = V -\frac{\partial\psi}{\partial  t}, \quad  A'_j = A_j + \frac{\partial\psi}{\partial  x_j}\quad (j= 1,2,\dots,d).
\end{equation}
We write \eqref{2.1} for this $(V',A')$  as $K'_{\Delta}(t,0)f$.  Then we have the formula 
\begin{equation}  \label{2.16}
     K'_{\Delta}(t,0)f  = e^{i\psi (t,\cdot)}K_{\Delta}(t,0)\left(e^{-i\psi (0,\cdot)}f\right)
\end{equation}
for  all $f \in L^2$, and we have the analogous relation between the limits $K'(t,0)f$ and $K(t,0)f$ as in \cite{Ichinose 1999}.
\end{thm}
   Next we consider the Lagrangian function  for the spin system
\begin{equation}  \label{2.17}
     \mathcal{L}_{s}(t,x,\dot{x}) = \mathcal{L}(t,x,\dot{x}) - H_{1}(t,x), 
\end{equation}
where $H_{1}(t,x) = (h_{1jk}(t,x);j\downarrow k\rightarrow 1,2,\dots,l)$ is a Hermitian matrix of degree $l$ and $\mathcal{L}(t,x,\dot{x})$  the Lagrangian function defined by \eqref{1.2}.  Then the corresponding quantized equation is given by
\begin{equation}  \label{2.18}
   i\frac{\partial u}{\partial t}(t) = \bigl[H(t)I + H_{1}(t)\bigr]u(t),
\end{equation}
where $u(t) = {}^t(u_1(t),\dots,u_l(t)) \in \mathbb{C}^l$, $H(t)$ is the operator defined by \eqref{1.3} and $I$ the identity matrix of degree $l$.
\par
    For a continuous path $q(\theta) \in \bR^{d} \ (s \leq \theta \leq t)$  let us define an $l\times l$ matrix $\mathcal{F}(\theta,s;q)\ (s \leq \theta \leq t)$ by the solution to
\begin{equation}  \label{2.19}
   \frac{d}{d \theta}\mathcal{A}(\theta)  = -iH_{1}(\theta,q(\theta))\mathcal{A}(\theta), \ \mathcal{A}(s) = I.
\end{equation}
Then, for the piecewise free moving path $\qdelta(\theta;x^{(0)},x^{(1)},
\dots,x^{(\nu-1)},x)$  we define the probability amplitude by
\begin{equation}  \label{2.21}
  \exp *iS_{s}(t,0;\qdelta) = \bigl(\exp iS(t,0;\qdelta)\bigr)\mathcal{F}(t,0;q_{\Delta}),
\end{equation}
using $S(t,s;q)$ defined by \eqref{1.4}.  Let $f = {}^{t}(f_{1},f_{2},\dots,f_{l}) \in C^{\infty}_0(\bR^d)^{l}$.
Then we define the approximation $K_{s\Delta}(t,0)f$ of the Feynman path integral $K_{s}(t,0)f$ for the system \eqref{2.17}  by replacing $e^{iS(t,0;\qdelta)}$ in \eqref{2.1} with $e^{*iS_{s}(t,0;\qdelta)}$ as in \cite{Ichinose 2007}.
\begin{thm}
Besides Assumptions 2.1 and 2.2 we assume
\begin{equation}  \label{2.22}
  |\partial_{x}^{\alpha} h_{1jk}(t,x)| \leq C_{\alpha},\ j,k = 1,2,\dots,l
\end{equation}
for all $\alpha$. Let $\rho^{*} > 0$ be the constant in Theorem 2.1.  Then we get the same assertions for $K_{s\Delta}(t,0)f$ as for $\kdelta(t,0)f$ in Theorem 2.1 with another constant $K \geq 0$, where $K_{s}(t,0)f = \lim_{|\Delta| \to 0}K_{s\Delta}(t,0)f \in C^{0}_{t}([0,T];(L^{2})^{l})$ for $f \in (L^{2})^{l}$ is the unique solution  in $C^{0}_{t}([0,T];(L^{2})^{l})$ to \eqref{2.18} with $u(0) = f$.
\end{thm}
\begin{rem}
Since we  see from \eqref{2.19} that $e^{iS(t,0;q)}\mathcal{F}(t,0;q)$ is the solution to
\begin{equation*}  
   \frac{d}{d t}\mathcal{U}(t)  = i\mathcal{L}_{s}(t,q(t),\dot{q}(t))\mathcal{U}(t), \ \mathcal{U}(0) = I,
\end{equation*}
we can write $\exp *iS_{s}(t,0;\qdelta)$ formally as $\exp i \displaystyle{\int_0^t}\mathcal{L}_{s}(\theta,\qdelta(\theta),\dot{q}_{\Delta}(\theta)d\theta.$  This is the reason why we express the right-hand side of \eqref{2.21} as $\exp *iS_{s}(t,0;\qdelta)$.
\end{rem}
\begin{rem}
We write
\begin{equation}  \label{2.20}
  q^{t,s}_{x,y}(\theta) = y + \frac{\theta-s}{t-s}(x-y), \ s\leq \theta \leq t
\end{equation}
for $x$ and $y$ in $\bR^{d}$ when $s \not= t$.  Then from Lemma 2.1 of \cite{Ichinose 2007} we have
\begin{equation*}
\mathcal{F}(t,0;\qdelta) = \mathcal{F}(t,\tau_{\nu-1};q^{t,\tau_{\nu-1}}_{x,x^{(\nu-1)}})\mathcal{F}(\tau_{\nu-1},\tau_{\nu-2};q^{\tau_{\nu-1},\tau_{\nu-2}}_{x^{(\nu-1)},x^{(\nu-2)}}) \cdots 
  \mathcal{F}(\tau_{1},0;q^{\tau_{1},0}_{x^{(1)},x^{(0)}}).
\end{equation*}
\end{rem}
\begin{rem}
Letting $M_* = 0$, we assume \eqref{2.8}, \eqref{2.10} and \eqref{2.22}.  Let $(V(t,x),A(t,x))$ be an arbitrary potential such that $V, \partial V/\partial x_j, \partial A_j/\partial t$ and $\partial A_j/\partial x_k\ (j, k = 1,2,\dots,d)$ are continuous in $\domain$.  Then we have proved in \cite{Ichinose 1999} and \cite{Ichinose 2007} the same assertions as in Theorems 2.1 and 2.2 .  Aside from this, letting $M_*=0$, we assume \eqref{2.8}, \eqref{2.9}, \eqref{2.12}, \eqref{2.22},
\begin{equation}  \label{2.23}
  |\partial_{x}^{\alpha} V(t,x)| \leq C_{\alpha}<x>,\  |\alpha| \geq 1
\end{equation}
and 
\begin{equation}  \label{2.23.2}
  |\partial_{x}^{\alpha}\partial_t V(t,x)| \leq C_{\alpha}<x>^M,\  |\alpha| \geq 1
\end{equation}
for a constant $M\geq 0$.  Using \eqref{1.1}, from \eqref{2.8} and \eqref{2.23} we have
\begin{equation}  \label{2.23.3}
  |\partial_{x}^{\alpha}\partial_t A_j(t,x)| = |\partial^{\alpha}_xE_j(t,x)| + |\partial^{\alpha}_x\partial_{x_j}V(t,x)| \leq C_{\alpha}<x>
\end{equation}
for all $\alpha$.  We note (3.3) in \cite{Ichinose 2017} or \eqref{5.11} in the present paper.  Then, under the assumptions above we can prove
 the same assertions as in Theorems 2.1 and 2.2 as in the proofs of the theorems stated in \cite{Ichinose 1999} and \cite{Ichinose 2007}.
\end{rem}
   In the end we will consider the multi-particle system.  For simplicity we will consider the 4-particle system
   \begin{align} \label{2.24}
   \mathcal{L}^{\sharp}(t,x,\dot{x}) =&\sum_{l=1}^{4}\Big\{\frac{m_{l}}{2}|\dot{x}(l)|^{2} + \dot{x}(l)\cdot A^{(l)}(t,x(l))  -V_{l}(t,x(l))\Big\} \notag \\
   & -2 \sum_{1 \leq j < k\leq 4}V_{jk}(t,x(j) -x(k)),
   \end{align}
   where $x = (x(1),x(2),x(3),x(4)) \in \bR^{4d}$.  The corresponding Schr\"odinger equation is given by
\begin{align} \label{2.25}
&  i\frac{\partial u}{\partial t}(t)  = \Biggl[\sum_{l=1}^{4}\left\{\frac{1}{2m_{l}}
      \left|\frac{1}{i}\frac{\partial}{\partial x(l)} - A^{(l)}(t,x(l))\right|^2 + V_{l}(t,x(l))\right\}\notag
      \\
      & + 2 \sum_{1 \leq j < k\leq 4}V_{jk}(t,x(j) -x(k)) \Biggr]u(t).
\end{align}
\vspace{0.2cm}
\par
{\bf Assumption 2.3}. (1) Each $\bigl(V_{l}(t,x(l)),A^{(l)}(t,x(l))\bigr)\ (l = 1,2)$ satisfies Assumption 2.1 with $M_{*} = M_{l*} > 0$. (2)  Each $V_{l}(t,x(l))\ (l = 3,4)$ satisfies \eqref{2.23} and \eqref{2.23.2}.
\vspace{0.2cm}
\par
	We define $E^{(l)}(t,x(l))$ and $B^{(l)}(t,x(l))\ (l= 1,2,3,4)$ by \eqref{1.1} where $A = A^{(l)}$ and $V = V_{l}$.
	\vspace{0.2cm}
\par
{\bf Assumption 2.4}.  Let $M_{l*}\ (l= 1,2)$ be the constants in Assumption 2.3 and $M_{l*} = 0\ (l = 3,4)$. (1)  $A^{(l)}$ and $(E^{(l)},B^{(l)})\ (l = 1,2,3,4)$ satisfy Assumption 2.2 with $M_{*} = M_{l*}$. (2) $V_{jk}(t,z)\ (z \in \bR^{d})$ satisfies 
\begin{equation}  \label{2.26}
  |\partial_{z}^{\alpha} V_{jk}(t,z)| \leq C_{\alpha},\  |\alpha| \geq 2.
\end{equation}
for all $1 \leq j < k \leq 4.$
\vspace{0.2cm}
\par
 We define the approximation $\kdelta^{\sharp}(t,0)f$ of the Feynman path integral $K^{\sharp}(t,0)f$ for the 4-particle system \eqref{2.24}  in the same way as \eqref{2.1}.
  \begin{thm}
  Suppose Assumptions 2.3 and 2.4.  Then we have the same assertions for $\kdelta^{\sharp}(t,0)f$ as for $\kdelta(t,0)f$ in Theorem 2.1 with other constants $\rho^{*} > 0$ and $K \geq 0$, where the Feynman path integral $K^{\sharp}(t,0)f \in C^{0}_{t}([0,T];L^{2})$ for $f \in L^{2}(\bR^{4d})$ is the unique solution 
   in  $C^{0}_{t}([0,T];L^{2})$ to \eqref{2.25} with $u(0) = f$.
  \end{thm}
  	Let us consider the spin system. Taking a Hermitian matrix $H_{1}(t,x) = (h_{1jk}(t,x);j\downarrow k\rightarrow 1,2,\dots,l_{0})\ (x \in \bR^{4d})$ of degree $l_{0}$ and using $\mathcal{L}^{\sharp}(t,x,\dot{x})$ defined by \eqref{2.24}, we determine
   \begin{equation} \label{2.27}
   \mathcal{L}^{\sharp}_{s}(t,x,\dot{x}) =\mathcal{L}^{\sharp}(t,x,\dot{x}) - H_{1}(t,x).
   \end{equation}
   For a path $q(\theta) \in \bR^{4d}\ (s \leq \theta \leq t)$  we define $\mathcal{F}^{\sharp}(\theta,s;q)$ by the solution to \eqref{2.19}.  Then we define $\exp *iS^{\sharp}_{s}(t,0;\qdelta)$ by \eqref{2.21} and $K^{\sharp}_{s\Delta}(t,0)f$ for $f \in C^{\infty}_0(\bR^{4d})^{l_0}$ in the same way as we did $K_{s\Delta}(t,0)f$.
   \begin{thm}
   Besides Assumptions 2.3 and 2.4 we assume \eqref{2.22}.  Let $\rho^{*} > 0$ be the constant in Theorem 2.3.  Then we have the same assertions for $K^{\sharp}_{s\Delta}(t,0)f$ as for $K_{s\Delta}(t,0)f$ in Theorem 2.2 with another constant $K \geq 0.$
   \end{thm}
\section{Examples}
In this section we will give some examples satisfying Assumptions 2.1 and 2.2 in \S 2.
\begin{lem}
Let $f \in C^{2}([0,\infty))$ and set
\begin{equation} \label{3.1}
V(x) = f(|x|^{2}), \ x \in \bR^{d}.
\end{equation}
Let $x \not= 0$ be an arbitrary point in $\bR^{d}$.  Then there exists an orthogonal matrix $\mathfrak{R}$ such that
\begin{align} \label{3.2}
&\frac{\partial^{2} V}{\partial x^{2}}(x) := \left( \frac{\partial^{2} V}{\partial x_{i} \partial x_{j}};i\downarrow j \rightarrow 1,2,\dots,d\right) \notag \\
& = 2 {}^{t} \mathfrak{R}f'(|x|^{2})\mathfrak{R} + 4{}^{t}\mathfrak{R}
 \left(
    \begin{array}{cccc}
      |x|^{2}f''(|x|^{2})& 0 & \ldots & 0 \\
      0 & 0 & \ldots & 0 \\
      \vdots & \vdots & \ddots & \vdots \\
      0 & 0 & \ldots & 0
    \end{array}
  \right)\mathfrak{R}.
\end{align}
\end{lem}
\begin{proof}
Let $\mathfrak{R} = (R_{ij};i\downarrow j\rightarrow 1,2,\dots,d)$ be an orthogonal matrix.  Then we have
\begin{equation} \label{3.3}
V(\dR x) = f(|\dR x|^{2}) = f(|x|^{2}) = V(x),
\end{equation}
which shows 
\begin{equation*} 
\frac{\partial V}{\partial x_{j}}(x) = \frac{\partial }{\partial x_{j}}V(\dR x) = \sum_{k=1}^{d}\frac{\partial V}{\partial x_{k}}(\dR x) R_{kj}
\end{equation*}
and so
\begin{equation*} 
\frac{\partial^{2} V}{\partial x_{i}\partial x_{j}}(x) = \sum_{k,l=1}^{d}
R_{li}\frac{\partial^{2} V}{\partial x_{l}\partial x_{k}}(\dR x)R_{kj}.
\end{equation*}
Hence 
\begin{equation} \label{3.4}
\frac{\partial^{2} V}{\partial x^{2}}(x) = {}^{t}\dR \frac{\partial^{2} V}{\partial x^{2}}(\dR x)\dR.
\end{equation}
\par
	On the other hand, from  \eqref{3.1} we see
\begin{equation*} 
\frac{\partial V}{\partial x_{j}}(x) = 2x_{j}f'(|x|^{2})
\end{equation*}
and so
\begin{equation*} 
\frac{\partial^{2} V}{\partial x_{i}\partial x_{j}}(x) = 2\delta_{ij}f'(|x|^{2}) + 4x_{i}x_{j}f''(|x|^{2}).
\end{equation*}
Consequently, letting $\overrightarrow{e_{1}} = (1,0,\dots,0) \in \bR^{d}$, we have
\begin{equation} \label{3.5}
\frac{\partial^{2} V}{\partial x_{i}\partial x_{j}}(|x|\overrightarrow{e_{1}}) = 2\delta_{ij}f'(|x|^{2}) + 4\delta_{1i}\delta_{1j}|x|^{2}f''(|x|^{2}).
\end{equation}
Let $x \not= 0$ be an arbitrary  point in $\bR^{d}$.  Then we can take an orthogonal matrix $\dR$ such that $\dR x = |x|\overrightarrow{e_{1}}$.  Then from \eqref{3.4} we have 
\begin{equation*} 
\frac{\partial^{2} V}{\partial x^{2}}(x) = {}^{t}\dR \frac{\partial^{2} V}{\partial x^{2}}(|x|\overrightarrow{e_{1}})\dR
\end{equation*}
and hence have \eqref{3.2} from \eqref{3.5}.
\end{proof}
	From Lemma 3.1 we can easily get the following.
	\begin{cor}
	Let $f \in C^{2}([0,\infty))$ such that
\begin{equation} \label{3.6}
f''(\theta) \geq 0\ (0 \leq \theta < \infty).
\end{equation}
We define $V(x)$ by \eqref{3.1}.  Then we have
\begin{equation} \label{3.7}
\frac{\partial^{2} V}{\partial x^{2}}(x) \geq 2f'(|x|^{2})I.
\end{equation}
	\end{cor}
	\begin{pro}
	Let $f \in C^{2}([0,\infty))$ satisfying \eqref{3.6} and $f'(0) \geq 0$.  Let $\mathfrak{A}$ be a regular real matrix.  We denote the smallest eigenvalue of ${}^{t}\dA\dA$ by $\beta > 0$.  We set
\begin{equation} \label{3.8}
V(x) = f(|\dA x|^{2}).
\end{equation}
Then we have
\begin{equation} \label{3.9}
\frac{\partial^{2} V}{\partial x^{2}}(x) \geq 2\beta f'(|\dA x|^{2})I \geq 2\beta  f'(\beta | x|^{2})I.
\end{equation}
\end{pro}
\begin{proof}
Setting $W(x) = f(|x|^{2})$, we have $V(x) = W(\dA x),$ which shows 
\begin{equation*} 
\frac{\partial^{2} V}{\partial x^{2}}(x) = {}^{t}\dA \frac{\partial^{2} W}{\partial x^{2}}(\dA x)\dA
\end{equation*}
as in the proof of \eqref{3.4}.  Hence
\[
\left(\frac{\partial^{2} V}{\partial x^{2}}(x)u,u\right)
= \left(\frac{\partial^{2} W}{\partial x^{2}}(\dA x)\dA u,\dA u\right)
\]
for $u \in \bR^{d}$, where $(\cdot,\cdot)$ is the inner product in $\bR^{d}$.  Since we have $\partial^{2}W(x)/\partial x^{2} \geq 2f'(|x|^2)I$ from \eqref{3.7}, we have
\[
\left(\frac{\partial^{2} V}{\partial x^{2}}(x)u,u\right)
\geq 2\left(f'(|\dA x|^{2})\dA u,\dA u\right),
\]
which leads to 
\[
\left(\frac{\partial^{2} V}{\partial x^{2}}(x)u,u\right)
\geq 2f'(|\dA x|^{2})\left({}^{t}\dA\dA  u, u\right) \geq 2\beta f'(|\dA x|^{2})|u|^{2}
\]
because of $f'(\theta) \geq 0\ (0 \leq \theta < \infty)$.  Hence we obtain the first inequality of \eqref{3.9}.  The second inequality follows from the fact that $f'(\theta)$ is an increasing function and $|\dA x|^{2} = ({}^{t}\dA \dA x,x) \geq \beta|x|^{2}$.
\end{proof}
\begin{exmp}
Let $(V,A)$ be the potential defined by \eqref{1.6} and \eqref{1.7} with an integer $M \geq 0$, real-valued $a_{\alpha}(t)$ and  $b_{j\alpha}(t)$ in $C^{1}([0,T])$.  I will prove that this $(V,A)$ satisfies Assumptions 2.1 and 2.2 with the integer $M_{*} = M$.
\par
	Noting \eqref{1.1}, we can easily see that we have only to prove \eqref{2.7}. Letting $f(\theta) = \theta^{M+1}$ in Corollary 3.2, we have
\begin{equation} \label{3.10}
\frac{\partial^{2}}{\partial x^{2}}|x|^{2(M+1)} \geq 2(M+1)|x|^{2M}I.
\end{equation}
Hence we can prove \eqref{2.7}, because we have
\begin{equation*} 
- \frac{\partial E}{\partial x} =\frac{\partial^{2} A}{\partial t\partial x }+ 
\frac{\partial^{2} V}{\partial x^{2}} = \frac{\partial^{2}}{\partial x^{2}}|x|^{2(M+1)} + O(<x>^{2M-1})
\end{equation*}
from \eqref{1.1}, \eqref{1.6} and \eqref{1.7}.
\end{exmp}
\begin{exmp}
Let $\dA (t)$ be a regular real matrix whose components are continuously differentiable on $[0,T]$.  We set
\begin{equation} \label{3.11}
V(t,x) = |\dA (t)x|^{2(M+1)} + V_{1}(t,x)
\end{equation}
with an integer $M \geq 0$, where we assume
\begin{equation} \label{3.12}
|\partial_{x}^{\alpha}V_{1}(t,x)| \leq C_{\alpha}<x>^{2(M+1) -|\alpha| - \delta_{\alpha}}
\end{equation}
for all $\alpha$ with constants $\delta_{\alpha} > 0$ and 
\begin{equation} \label{3.13}
|\partial_{x}^{\alpha}\partial_{t}V_{1}(t,x)| \leq C_{\alpha}<x>^{2(M+1)}
\end{equation}
for all $\alpha$.  In addition, we assume 
\begin{equation} \label{3.14}
|\partial_{x}^{\alpha}A_{j}(t,x)| \leq C_{\alpha}<x>^{M+1-|\alpha| -\delta_{\alpha}}\quad (j = 1,2,\dots,d)
\end{equation}
for all $\alpha$ and 
\begin{equation} \label{3.15}
|\partial_{x}^{\alpha}\partial_{t}A_{j}(t,x)| \leq C_{\alpha}<x>^{M+1-|\alpha|}\quad (j = 1,2,\dots,d)
\end{equation}
for all $\alpha$.  Then this potential $(V,A)$ satisfies Assumptions 2.1 and 2.2 with $M_{*} = M$.
\par
In fact we have only to prove \eqref{2.7} as in the arguments in Example 3.1.  Letting $f(\theta) = \theta^{M+1}$ in Proposition 3.3, we have
\begin{equation} \label{3.16}
\frac{\partial^{2}}{\partial x^{2}}|\mathfrak{A}(t)x|^{2(M+1)} \geq 2(M+1)\beta|\mathfrak{A}(t)x|^{2M} \geq 2(M+1)\beta^{M+1}|x|^{2M}
\end{equation}
with $\beta > 0$.  Hence we can prove \eqref{2.7} as in the proof of Example 3.1.
\end{exmp}
\begin{exmp}
Let $\dA(t)$ be the matrix in Example 3.2.  We set
\begin{equation} \label{3.17}
V(t,x) = \left( 1 + |\dA(t)x|^{2}\right)^{M+1} + V_{1}(t,x)
\end{equation}
with a constant $M \geq 0$, where $V_{1}(t,x)$ is assumed to satisfy \eqref{3.12} and \eqref{3.13}.
Suppose that $A(t,x)$ satisfies \eqref{3.14} and \eqref{3.15}.  In addition,  when $M$ in \eqref{3.17} is in $(0,1)$, we assume \eqref{3.14} with $M = 0$. Then this potential $(V,A)$ satisfies Assumptions 2.1 and 2.2.  In fact we have only to prove \eqref{2.7}.  Letting $f(\theta) = (1+\theta)^{M+1}$ in Proposition 3.3, we have
\begin{align} \label{3.18}
& \frac{\partial^{2}}{\partial x^{2}}\left(1 + |\mathfrak{A}(t)x|^{2} \right)^{M+1}\geq 2(M+1)\beta\left(1 + |\mathfrak{A}(t)x|^{2} \right)^{M}  \notag \\
& \geq 2(M+1)\beta\left(1 + \beta|x|^{2} \right)^{M} \geq 2(M+1)\beta^{M+1}|x|^{2M}.
\end{align}
Hence we can prove \eqref{2.7} as in the proof of Example 3.2.
\end{exmp}
\section{The Feynman path integrals for bosons and fermions}
In this section we consider the quantum spin system consisting of $N$ particles.  We write $x = (\bx_1,\bx_2,\dots,\bx_N) \in \bR^{3N}$.  The Lagrangian function is given by %
 \begin{align} \label{4.1}
& {\cal L}^{\sharp}_{is}(t,x,\dot{x}) =  \sum_{i=1}^N 
 \biggl\{\frac{m_i}{2}|\dot{\bx}_i|^2 + e_i\dot{\bx}_i\cdot \bA_i(t,\bx_i) -e_iV_i(t,\bx_i) + I_1\otimes \cdots \otimes I_{i-1} \notag \\
 & \otimes \frac{e_i}{m_i}\bB_i(t,\bx_i)\cdot \bs_i\otimes I_{i+1}\otimes \cdots\otimes I_N
 \biggr\} - \sum_{j,k=1,j \not= k}^Ne_je_kV_{jk}(t,\bx_j-\bx_k) \notag \\
 &  \equiv {\cal L}^{\sharp}_{i}(t,x,\dot{x}) + \sum_{i=1}^N I_1\otimes \cdots \otimes I_{i-1} \otimes \frac{e_i}{m_i}\bB_i(t,\bx_i)\cdot \bs_i\otimes I_{i+1}\otimes \cdots\otimes I_N
  \end{align}
in terms of the tensor product, where $\bA_i \in \bR^3, \bB_i \in \bR^3, V_i \in \bR, V_{jk} \in \bR, $ $\bs_i = 
(\hat{s}_1,\hat{s}_2,\hat{s}_3)$ are spin matrices with three components and $I_j$ the identity matrix for the j-th particle.  In particular we suppose that all particles are identical.  Hence we suppose $m_i = m, e_i = e, \bA_i = \bA, \bB_i = \bB, V_i = V, V_{jk} = W$ and $\bs_i = \bs$.  Let $L$ be the magnitude of spin of particles.
We note that the N-fold tensor product $L^{2}(\bR^{3})^{2L+1}\otimes \cdots \otimes L^{2}(\bR^{3})^{2L+1}$ is isomorphic to $ L^{2}(\bR^{3N})^{l}$ with $l = (2L+1)^{N}$ (cf. Theorem II.10 on p. 52 in \cite{Reed-Simon I}), which we write as $\mathfrak{H}$.
The  Schr\"odinger equation for the Lagrangian \eqref{4.1} is given by
\begin{align} \label{4.2}
& i \frac{\partial u}{\partial t}(t)  =\Biggl[\sum_{j=1}^N \Bigg\{\frac{1}{2m}
      \left|\frac{1}{i}\frac{\partial}{\partial \bx_j} - e\bA(t,\bx_j)\right|^2 + eV(t,\bx_j) - I_1\otimes \cdots \otimes I_{j-1}\notag\\ 
&   \otimes \frac{e}{m}\bB(t,\bx_j)\cdot \bs\otimes I_{j+1}\otimes \cdots\otimes I_N\Biggr\} + e^2 \sum_{j,k=1,j \not= k}^N W(t,\bx_j-\bx_k)
\Biggr]u(t).
\end{align}
We note that if $W = 0$ and $u(t,x) = u_1(t,\bx_1)\otimes \cdots \otimes u_N(t,\bx_N)$, \eqref{4.2} is written as
\begin{equation*} 
 0 = \sum_{j=1}^Nu_1(t)\otimes \cdots u_{j-1}(t)\otimes \left[i\frac{\partial}{\partial t} - H_j(t)\right]u_j(t)\otimes \cdots u_N(t),
\end{equation*}
where $H_j(t) = |i^{-1}\partial_{\bx_j} - e\bA(t,\bx_j)|^2/(2m) + eV(t,\bx_j) - e\bB(t,\bx_j)\cdot\bs/m$.
\par
Let $S^{\sharp}_i(t,s;q)$ be the classical action for $\mathcal{L}^{\sharp}_i(t,x,\dot{x})$ defined by \eqref{4.1}.
We define the approximation $K^{\sharp}_{is\Delta}(t,0)f$ of the Feynman path integral $K^{\sharp}_{is}(t,0)f$ for \eqref{4.1} in the same way as we did $K^{\sharp}_{s\Delta}(t,0)f$ before Theorem 2.4, where
$f = \bigl\{f(\bx_1,s_1,\bx_2,s_2,\dots,\bx_N,s_N); s_j = -L, -L+1,\dots, L\ (j = 1,2,\dots,N)\bigr\} \in C^{\infty}_0(\bR^{3N})^l$.
 That is, we define $\mathcal{F}^{\sharp}_{i}(\theta,s;q)$
for a path $q(\theta) \in \bR^{3N}\ (s \leq \theta \leq t)$ by the solution
\begin{equation*}
\frac{d}{d\theta}\mathcal{A}(\theta) = -iH_1(\theta,q(\theta))\mathcal{A}(\theta),\  \mathcal{A}(s) = I,
\end{equation*}
where $H_1(t,x) = - \sum_{j=1}^NI_1\otimes \cdots \otimes I_{j-1}
   \otimes e\bB(t,\bx_j)\cdot \bs/m\otimes I_{j+1}\otimes \cdots\otimes I_N $.  Next we define the probability amplitude by \eqref{2.21} and eventually define $K^{\sharp}_{is\Delta}(t,0)f$ by \eqref{2.1}. 
   \par
  We define $\mathcal{F}(\theta,s;\mathbf{q})\ (s \leq \theta \leq t)$ for a continuous path $\bq(\theta) \in \bR^{3}\ (s \leq \theta \leq t)$ by the solution 
\begin{equation} \label{4.3}
\frac{d}{d\theta}\mathcal{A'}(\theta) = i\frac{e}{m}\bigl(\bB(\theta,\bq(\theta))\cdot \bs\bigr)\mathcal{A'}(\theta),\  \mathcal{A'}(s) = I.
\end{equation}
Then we can easily have
\begin{align*}
& \frac{d}{d\theta}\, \mathcal{F}(\theta,s;\mathbf{q}_1) \otimes \cdots \otimes  \mathcal{F}(\theta,s;\mathbf{q}_N) 
 = \sum_{j=1}^N \mathcal{F}(\theta,s;\mathbf{q}_1)\otimes \cdots \otimes \mathcal{F}(\theta,s;\mathbf{q}_{j-1}) \\
&  \otimes \frac{d}{d\theta}\, \mathcal{F}(\theta,s;\mathbf{q}_j) \otimes \mathcal{F}(\theta,s;\mathbf{q}_{j+1})\otimes \cdots  \mathcal{F}(\theta,s;\mathbf{q}_N) = \sum_{j=1}^N \Bigl(I_1 \otimes \cdots \otimes i\frac{e}{m}\mathbf{B}(\theta,\mathbf{q}_j(\theta))\cdot\bs \\
& \otimes I_{j+1} \otimes \cdots \otimes I_N \Bigr)\mathcal{F}(\theta,s;\mathbf{q}_1) \otimes \cdots \otimes  \mathcal{F}(\theta,s;\mathbf{q}_N) \\
& = -iH_1(\theta,q(\theta))\mathcal{F}(\theta,s;\mathbf{q}_1) \otimes \cdots \otimes  \mathcal{F}(\theta,s;\mathbf{q}_N).
\end{align*}
Hence 
we have
\begin{equation} \label{4.4}
\mathcal{F}^{\sharp}_{i}(\theta,s;q) = \mathcal{F}(\theta,s;\bq_{1})\otimes \cdots \otimes \mathcal{F}(\theta,s;\bq_{N})
\end{equation}
because of the uniqueness of  solutions to the ordinary differential equation, where $\mathbf{B}(t,\bx)$ are assumed to be continuous in $[0,T]\times \bR^{3}$.
\par
   Let $\widehat{P}_{ij}\ (i, j = 1,2,\dots,N)$ be the operator exchanging the $i$-th particle  and the $j$-th one.  That is,  we define 
\begin{align} \label{4.5}
& \widehat{P}_{ij}\bigl( f_1(\bx_1)\otimes \cdots \otimes f_i(\bx_i)\otimes \cdots \otimes f_j(\bx_j) \otimes \cdots \otimes f_N(\bx_N)\bigr) \notag \\
& =  f_1(\bx_1)\otimes \cdots \otimes f_j(\bx_i)\otimes \cdots \otimes f_i(\bx_j) \otimes \cdots \otimes f_N(\bx_N)
\end{align}
for $f_j(\bx_j) \in L^2(\bR^3)^{2L+1}\ (j = 1,2,\dots,N)$ and extend $\widehat{P}_{ij}$ for $f = \sum_{n=1}^{\infty}f_1^{(n)}(\bx_1)\otimes \cdots \otimes f_N^{(n)}(\bx_N) \in \mathfrak{H}$.
   \par
   The following theorem shows that the Feynman path integrals $K^{\sharp}_{is}(t,0)f$ are expressing bosons and fermions.
   \begin{thm}
   Assume that $(V(t,\bx),\bA(t,\bx))\ (\bx \in \bR^{3}), \bB(t,\bx)$ and $W(t,\bx)$ satisfy Assumptions 2.1-2.2, \eqref{2.22} and \eqref{2.26}, respectively.  Then we have: (1) The same assertions for $K^{\sharp}_{is\Delta}(t,0)f$ as for $\kdelta(t,0)f$ in Theorem 2.1 hold, where $K^{\sharp}_{is}(t,0)f \in C^{0}_{t}([0,T];\mathfrak{H})\ (f \in \mathfrak{H})$ is the unique solution  in  $C^{0}_{t}([0,T];\mathfrak{H})$ to \eqref{4.2} with $u(0) = f$. 
   (2) If $f  \in \mathfrak{H}$ is symmetric, i.e. $\hat{P}_{ij}f = f$ for all $i$ and $j$, so is $K^{\sharp}_{is}(t,0)f$. (3) If $f \in \mathfrak{H}$ is antisymmetric , i.e. $\hat{P}_{ij}f = - f$ for all $i \not= j$, so is $K^{\sharp}_{is\Delta}(t,0)f$. 
   \end{thm}
  \begin{proof}
  The first assertion (1) follows from Theorem 2.4. 
  Let's prove the second assertion.  For simplicity suppose $N = 2$.  Let $c(\bx_1,\bx_2) \in \mathbb{C}$ be a bounded measurable  function.  Then we can prove
\begin{equation} \label{4.6}
\hat{P}_{12}\, c(\bx_1,\bx_2)f_1(\bx_1)\otimes f_2(\bx_2) = c(\bx_2,\bx_1)f_2(\bx_1)\otimes f_1(\bx_2) 
\end{equation}
from the definition of $\hat{P}_{12}$, approximating $c(\bx_1,\bx_2)$ by $\sum_{j=1}^n c_1^{(j)}(\bx_1)c_2^{(j)}(\bx_2)$. 
\par
 Let $0 \leq t - s \leq \rho^*$.
 Setting $\bq_j(\theta) := \bq^{t,s}_{\bx_j,\by_j}(\theta)\ (s \leq \theta \leq t, j = 1,2)$, we write 
\begin{align} \label{4.7}
& \mathcal{C}^{\sharp}_{is}(\ts)(f_1\otimes f_2) & \notag \\
&
:= \iint e^{iS^{\sharp}_i(t,s;\bq_1,\bq_2)}\mathcal{F}^{\sharp}_i(t,s;\bq_{1},\bq_{2}) 
 \bigl(f_1(\bq_1(s))\otimes f_2(\bq_2(s)) \bigr)d\by_1d\by_2
\end{align}
for $f_j(\bx_j) \in C^{\infty}_0(\bR^3)^{2L+1}\ (j=1,2)$, which belongs to $L^2(\bR^6)^l$ from (1) of Theorem 2.4.  
From \eqref{4.4} we can write
\begin{align*} 
& \mathcal{C}^{\sharp}_{is}(\ts)(f_1\otimes f_2) = \iint e^{iS^{\sharp}_i(t,s;\bq_1,\bq_2)}\mathcal{F}(t,s;\bq_{1})f_1(\bq_1(s))\\
& \quad  \otimes  \mathcal{F}(t,s;\bq_{2})  f_2(\bq_2(s))d\by_1d\by_2.
\end{align*}
Making the same arguments as in the proof of \eqref{4.6}, by the exchange of $(\bx_1,\by_1)$ and $(\bx_2,\by_2)$ in the above equation we can prove
\begin{align*} 
&
 \widehat{P}_{12}\Bigl(\chi(\epsilon\bx_1)\chi(\epsilon\bx_2)\mathcal{C}^{\sharp}_{is}(\ts)(f_1\otimes f_2)\Bigr) = \chi(\epsilon\bx_2)\chi(\epsilon\bx_1) \iint e^{iS^{\sharp}_i(t,s;\bq_2,\bq_1)} \\
& \times \mathcal{F}(t,s;\bq_{1})f_2(\bq_1(s)) \otimes  \mathcal{F}(t,s;\bq_{2})  f_1(\bq_2(s))d\by_1d\by_2.
\end{align*}
Letting $\epsilon \to 0$, we have
\begin{align} \label{4.8}
& \hat{P}_{12}\,\mathcal{C}^{\sharp}_{is}(\ts)(f_1\otimes f_2) = \iint e^{iS^{\sharp}_i(t,s;\bq_2,\bq_1)} \mathcal{F}(t,s;\bq_{1})f_2(\bq_1(s))
\notag \\
& \qquad \otimes  \mathcal{F}(t,s;\bq_{2})  f_1(\bq_2(s))d\by_1d\by_2.
 \end{align}
 Using $S^{\sharp}_i(t,s;\bq_1,\bq_2) = S^{\sharp}_i(t,s;\bq_2,\bq_1)$,
 we have
 \begin{align} \label{4.9}
&  \hat{P}_{12}\,\mathcal{C}^{\sharp}_{is}(\ts)(f_1\otimes f_2) = \iint e^{iS^{\sharp}_i(t,s;\bq_1,\bq_2)}\mathcal{F}(t,s;\bq_{1})\otimes  \mathcal{F}(t,s;\bq_{2})\notag \\
& \cdot\bigl(f_2(\bq_1(s))\otimes f_1(\bq_2(s)) \bigr)d\by_1d\by_2 = \iint e^{iS^{\sharp}_i(t,s;\bq_1,\bq_2)}\mathcal{F}(t,s;\bq_{1})  \notag\\
& \quad \otimes  \mathcal{F}(t,s;\bq_{2})\hat{P}_{12}(f_1\otimes f_2)d\by_1d\by_2 = \mathcal{C}^{\sharp}_{is}(\ts)\hat{P}_{12}(f_1\otimes f_2).
\end{align}
We have proved in (1) of Theorem 2.4 that $ \mathcal{C}^{\sharp}_{is}(\ts)$ is a bounded operator on $\dH$.  Hence from \eqref{4.9} we see
\begin{equation} \label{4.10}
 \hat{P}_{12}\,\mathcal{C}^{\sharp}_{is}(\ts)f = \mathcal{C}^{\sharp}_{is}(\ts)\hat{P}_{12}f
\end{equation}
for $f \in \dH$.
\par
Noting Remark 2.2, from \eqref{2.1} and \eqref{2.21}  we can write
\begin{equation*} 
 K^{\sharp}_{is\Delta}(t,0)f = \lim_{\epsilon\to 0+}
 \mathcal{C}^{\sharp}_{is}(t,\tau_{\nu-1})\chi(\epsilon\cdot)\cdots \chi(\epsilon\cdot)\mathcal{C}^{\sharp}_{is}(\tau_1,0)f
\end{equation*}
for $f \in C^{\infty}_0(\bR^{6})^l$.  Since we have
\begin{align*}
& \mathcal{C}^{\sharp}_{is}(t,\tau_{\nu-1})\chi(\epsilon\cdot)\mathcal{C}^{\sharp}_{is}(\tau_{\nu-1},\tau_{\nu-2})\chi(\epsilon\cdot)\cdots\chi(\epsilon\cdot)\mathcal{C}^{\sharp}_{is}(\tau_{1},0)f - \mathcal{C}^{\sharp}_{is}(t,\tau_{\nu-1})\\
&\cdot \mathcal{C}^{\sharp}_{is}(\tau_{\nu-1},\tau_{\nu-2}) \cdot\cdots\mathcal{C}^{\sharp}_{is}(\tau_{1},0)f = \sum_{j=1}^{\nu-1}\mathcal{C}^{\sharp}_{is}(t,\tau_{\nu-1})\chi(\epsilon\cdot)\mathcal{C}^{\sharp}_{is}(\tau_{\nu-1},\tau_{\nu-2})\chi(\epsilon\cdot)\cdots\\
& \cdot \chi(\epsilon\cdot)\mathcal{C}^{\sharp}_{is}(\tau_{j+1},\tau_{j})\bigl\{\chi(\epsilon\cdot) - 1\bigr\}\mathcal{C}^{\sharp}_{is}(\tau_{j},\tau_{j-1})\mathcal{C}^{\sharp}_{is}(\tau_{j-1},\tau_{j-2})\cdots \mathcal{C}^{\sharp}_{is}(\tau_{1},0)f,
\end{align*}
by (1) of Theorem 2.4 we obtain
\begin{equation} \label{4.11}
 K^{\sharp}_{is\Delta}(t,0)f = 
 \mathcal{C}^{\sharp}_{is}(t,\tau_{\nu-1}) \mathcal{C}^{\sharp}_{is}(\tau_{\nu-1},\tau_{\nu-2})\cdots \mathcal{C}^{\sharp}_{is}(\tau_1,0)f
\end{equation}
for $f \in C^{\infty}_0(\bR^{6})^l$ and so for $f \in \dH$.
Therefore, by \eqref{4.10} we have
\begin{equation} \label{4.12}
 \hat{P}_{12}K^{\sharp}_{is\Delta}(t,0)f = K^{\sharp}_{is\Delta}(t,0)\hat{P}_{12}f.
\end{equation}
 Since $\hat{P}_{12}f = f$ holds from the assumption, we obtain
\begin{equation*} 
 \hat{P}_{12}K^{\sharp}_{is\Delta}(\ts)f = K^{\sharp}_{is\Delta}(\ts)f.
\end{equation*}
This shows that $K^{\sharp}_{is\Delta}(t,0)f$ is symmetric, which completes the proof of the second assertion.
   In the same way the third assertion is proved from \eqref{4.12}.
  \end{proof}
  \begin{rem}
  We have supposed Assumptions 2.1 and 2.2 for $(V(t,\bx),\bA(t,x))$ in Theorem 4.1.  In place of these assumptions we suppose the assumptions stated in Remark 2.3 for $(V,\bA)$.  Then we can prove the same assertions as in Theorem 4.1 as in the proof of Theorem 4.1.
  \end{rem}
\section{Stability of $\mathcal{C}(t,s)$}
Let $\mathcal{L}(t,x,\dot{x})$ and $S(t,s;q)$ be the Lagrangian function and the classical action defined by \eqref{1.2} and\eqref{1.4}, respectively.  Let $\qts$ be the path defined by \eqref{2.20} and write 
\begin{equation} \label{5.1}
\gamma^{t,s}_{x,y}: \gamma^{t,s}_{x,y}(\theta) = (\theta,\qts(\theta)) \in \bR^{d+1}.
\end{equation}
Then we have
\begin{align}  \label{5.2}
   S(t,s;q^{t,s}_{x,y})
       & = \frac{m|x - y|^2}{2(t - s)} + \int_{\gamma^{t,s}_{x,y}} \bigl(A\cdot dx - Vdt\bigr)
        \notag \\
     & = \frac{m|x - y|^2}{2(t - s)} + (x - y)\cdot\int^1_0 A(s+ \theta\rho,
y+ \theta (x - y))d\theta \notag \\
   & \qquad -  \int^t_s  V(\theta, y+ \frac{\theta - s}{t-s} (x - y))d\theta \notag \\
   & = \frac{m|x - y|^2}{2(t - s)} + (x - y)\cdot\int^1_0 A(t - \theta\rho,
x - \theta (x - y))d\theta \notag \\
   & \qquad - \rho \int^1_0  V(t - \theta\rho,x - \theta (x - y))d\theta, \ \rho = t - s.
\end{align}
Let $M \geq 0$ and suppose that $p(x,w) \in \Cspace(\bR^{2d})$ satisfies 
\begin{equation} \label{5.3}
|\partial_{w}^{\alpha}\partial_{x}^{\beta}p(x,w)| \leq C_{\alpha\beta}<x;w>^M, \ (x,w) \in \bR^{2d}
\end{equation}
for all $\alpha$ and $\beta$, where $<x;w> = \sqrt{1 + |x|^2 + |w|^2}$.  We write the semi-norms of $\Sspace = \Sspace(\bR^d)$ as $|f|_l = \sum_{|\alpha + \beta| \leq l}\sup \bigl\{|x^{\alpha}\partial_{x}^{\beta}f(x)|; x \in \bR^d \bigr\}\ (l = 0,1,2,\dots)$.  For $f \in \Sspace$ we define
\begin{equation}  \label{5.4}
P(t,s)f =
        \begin{cases}
            \begin{split}
              & \sqrt{m/(2\pi i\rho)}^{\ d}
                  \int \bigl(\exp iS(t,s; q^{t,s}_{x,y})\bigr) \\
         &\hspace{2cm} \times    p(x,(x-y)/\sqrt{\rho})f(y)dy,
                      \end{split}
         & s < t ,
                  \\
\begin{split}
        & \sqrt{m/(2\pi i)}^{\ d}
    \text{Os}-\int (\exp im|w|^2/2)\\
             &\hspace{2cm}\times p(x,w)dwf(x).
             \end{split}
            & s = t
        \end{cases}
\end{equation}
 Then the formal adjoint operator $P(t,s)^{\dag}$ of $P(t,s)$ on $\Sspace$ is given by
\begin{equation}  \label{5.5}
P(t,s)^{\dag}f =
        \begin{cases}
            \begin{split}
              & \sqrt{i m/(2\pi \rho)}^{\ d}
                  \int \bigl(\exp -iS(t,s; q^{t,s}_{y,x})\bigr) \\
         &\hspace{2cm} \times    p(y,(y-x)/\sqrt{\rho})^*f(y)dy,
                      \end{split}
         & s < t ,
                  \\
\begin{split}
        & \sqrt{im/(2\pi )}^{\ d}
    \text{Os}-\int (\exp -im|w|^2/2)\\
             &\hspace{2cm}\times p(x,w)^{*}dwf(x),
             \end{split}
            & s = t.
        \end{cases}
\end{equation}
\par
We have the following from Lemma 2.1 of \cite{Ichinose 1999}.
\begin{lem}
We define $P(t,s)f$ by \eqref{5.4} for $f \in \Sspace$.  Assume \eqref{2.3} and \eqref{2.9}.  Then, $\partial_{x}^{\alpha}(P(t,s)f)$ are continuous in $0 \leq s \leq t \leq T$ and $x \in \bR^{d}$ for all $\alpha$.
\end{lem}
  Taking $1$ as $p(x,w)$ in \eqref{5.4}, for $f \in \Sspace$ we define
\begin{equation}  \label{5.6}
\mathcal{C}(t,s)f =
        \begin{cases}
            \begin{split}
              & \sqrt{m/(2\pi i\rho)}^{\ d}
                  \int \bigl(\exp iS(t,s; q^{t,s}_{x,y})\bigr)f(y)dy,  \\
                      \end{split}
         & s < t ,
                  \\
\begin{split}
        & f, 
             \end{split}
            & s = t.
        \end{cases}
\end{equation}
Using Lemma 5.1, we can write $\kdelta(t,0)f$ defined by \eqref{2.1} as 
\begin{equation}  \label{5.7}
     \kdelta(t,0)f = \limepsilon{\cal C}(t,\tau_{\nu-1})\chi(\epsilon\cdot){\cal C}(\tau_{\nu-1},\tau_{\nu-2})\chi(\epsilon\cdot)\cdots \chi(\epsilon\cdot)
{\cal C}(\tau_1,0)f
\end{equation}
for $f \in \Sspace$ under the assumptions of Lemma 5.1.
\begin{lem}
We assume that $\partial_{x}^{\alpha}V(t,x), \partial_{x}^{\alpha}A_{j}(t,x)$ and $\partial_{x}^{\alpha}\partial_{t}A_{j}(t,x)$ are continuous in $\domain$ for $|\alpha| \leq 1$.  Let $p(x,w)$ be a function satisfying \eqref{5.3}.  Then for any $0 < \epsilon \leq 1$ and $0 \leq s < t \leq T$  we have
\begin{align}  \label{5.8}
    & P(t,s)^{\dag}\chi(\epsilon\cdot)^2P(t,s)f = \left(\frac{m}{2\pi(t-s)}\right)^d\int f(y)dy \int\chi(\epsilon z)^2
    \notag \\
    &\times  \left(\exp i(x - y)\cdot \frac{m\Phi}{t-s} \right)p\left(z,\frac{z - x}{\sqrt{t-s}}\right)^*
     p\left(z,\frac{z - y}{\sqrt{t-s}}\right)dz,
\end{align}
\begin{equation}  \label{5.9}
  \Phi = \Phi(t,s;x,y,z) = (\Phi_1,\dots,\Phi_d),
\end{equation}
\begin{align}  \label{5.10}
    & \Phi_j =  z_j - \frac{x_j + y_j}{2} + \frac{t - s }{m}
              \int_0^1A_j(s,x + \theta(y - x))d\theta           \notag \\
      & -\frac{(t - s)^2 }{m}\int_0^1\int_0^1
                \sigma_1E_j(\tau(\sigma),\zeta(\sigma))
                       d\sigma_1d\sigma_2 \notag \\
                        & - \frac{t - s }{m}\sum_{k=1}^d (z_k - 
x_k)\int_0^1\int_0^1\sigma_1B_{jk}
                     (\tau(\sigma),\zeta(\sigma))d\sigma_1d\sigma_2  
\end{align}
or
\begin{align}  \label{5.11}
    & \Phi_j =  z_j - \frac{x_j + y_j}{2} + \frac{t - s }{m}
              \int_0^1A_j(s,x + \theta(y - x))d\theta           \notag \\
      & -\frac{(t - s)^2 }{m}\int_0^1\int_0^1
                \sigma_1E_j(\tau(\sigma),\zeta(\sigma))
                       d\sigma_1d\sigma_2  - \frac{(t - s)^2 }{m}\int_0^1d\theta \sum_{k=1}^d (z_k - 
x_k) \notag \\
                        &\times \int_0^1\int_0^1\sigma_1(1-\sigma_1)\frac{\partial B_{jk}}{\partial t}
                     (s+\theta(1-\sigma_1)\rho,\zeta(\sigma))d\sigma_1d\sigma_2, 
\end{align}
where
\begin{equation}   \label{5.12}
    \bigl(\tau(\sigma),\zeta(\sigma)\bigr)         
       = \bigl(t - \sigma_1(t - s), z + \sigma_1(x - z)
                  + \sigma_1\sigma_2(y - x)\bigr)\ \in \bR^{d+1}.
\end{equation}
\end{lem}
\begin{proof}
We have proved \eqref{5.8}, \eqref{5.9} and \eqref{5.10} in Proposition 3.3 of \cite{Ichinose 1997} and Lemma 2.2 of \cite{Ichinose 1999}.  So we will prove \eqref{5.8}, \eqref{5.9} and \eqref{5.11}, though these have been proved in Lemma 3.1 of \cite{Ichinose 2017} in essentials.  Let $\mathbf{\Delta}$ be the 2-dimensional plane with oriented boundary consisting of $-\{(s,y+ \theta(x - y)); 0 \leq \theta \leq 1\},  \{(s+ \theta\rho,y+ \theta(z - y)); 0 \leq \theta \leq 1\}$ and $-\{(s+\theta\rho,x+ \theta(z - x)); 0 \leq \theta \leq 1\}$.  Then we have
\begin{equation*}   
   \lim_{t\to s+0} \iint_{\mathbf{\Delta}}d(A\cdot x - V dt) = 0.
\end{equation*}
Hence from the proof of Lemma 3.2 in \cite{Ichinose 1997} we have
\begin{equation}   \label{5.13}
   \sum_{j=1}^d(x_j - y_j)\sum_{k=1}^d(z_k - x_k)\int_0^1\int_0^1\sigma_1B_{jk}(s,\zeta(\sigma))d\sigma_1d\sigma_2 = 0
\end{equation}
for all $x, y$ and $z$ in $\bR^d$.  Multiplying \eqref{5.13} by $(t-s)/m$ and adding this to $(x - y)\cdot \Phi$ where $\Phi$ is defined by \eqref{5.10}, we have \eqref{5.8}, \eqref{5.9} and \eqref{5.11}.
\end{proof}
\begin{lem}
Assume \eqref{2.7} and let $C_1 \geq 0$ be the constant in \eqref{2.7}.  Then, there exists a constant $C'_* > 0$ such that for all $X = (x,y,z) \in \bR^{3d}$ we have
\begin{align} \label{5.14}
& - \int_0^1\int_0^1 \sigma_1(1-\sigma_1)\biggl\{\frac{1}{2m}\left(\frac{\partial E}{\partial x} + \frac{{}^t\partial E}{\partial x}\right)(\tau(\sigma),\zeta(\sigma)) - C_1\biggr\}d\sigma_1d\sigma_2  \notag \\
&  \geq C'_*|X|^{2M_*}.
\end{align}
\end{lem}
\begin{proof}
For a while we write 
\begin{equation*}   
  Q(t,x) = -\frac{1}{2m}\left(\frac{\partial E}{\partial x}(t,x) + \frac{{}^t\partial E}{\partial x}(t,x) \right) + C_1
\end{equation*}
and $X' = (z,x - z,y-x) = |X'|(\omega'_1,\omega'_2,\omega'_3) \in \bR^{3d}$.  From \eqref{2.7} we have
\begin{align} \label{5.15}
& |X'|^{-2M_*} \int_0^1\int_0^1 \sigma_1(1-\sigma_1)Q(\tau(\sigma),z+\sigma_1(x - z)+\sigma_1\sigma_2(y-x))d\sigma_1d\sigma_2 \notag \\
& \geq C_* \int_0^1\int_0^1 \sigma_1(1-\sigma_1)|\omega'_1 + \sigma_1\omega'_2 + \sigma_1\sigma_2\omega'_3|^{2M_*}d\sigma_1d\sigma_2
\end{align}
for $X' \not= 0$. If the right-hand side of \eqref{5.15} is equal to zero for a point $(\omega'_1,\omega'_2,\omega'_3)$ such that $|\omega'_1|^2 + |\omega'_2|^2 + |\omega'_3|^2 = 1$, we have
\begin{equation*}   
  \omega'_1 + \sigma_1\omega'_2 + \sigma_1\sigma_2\omega'_3 = 0
  \end{equation*}
for all $0 < \sigma_1 < 1$ and $0 \leq \sigma_2 \leq 1$, which means $\omega'_1 = \omega'_2 = \omega'_3 = 0$.  This is contradiction.  Hence there exists a constant $C'_* > 0$ such that
\begin{equation*} 
 \int_0^1\int_0^1 \sigma_1(1-\sigma_1)Q(\tau(\sigma),z+\sigma_1(x - z)+\sigma_1\sigma_2(y-x))d\sigma_1d\sigma_2  \geq C'_* |X'|^{2M_*},
\end{equation*}
which shows \eqref{5.14} with another constant $C'_* > 0$ because of $|X'| \geq C|X|$ with a constant $C > 0$.
\end{proof}
For a while we write the constant $C_1$ in \eqref{2.7} and \eqref{5.14} as $a$.  Let us write 
\begin{align} \label{5.16}
& E'_0(t,s;x,y,z) = - \int_0^1\int_0^1 \sigma_1(1-\sigma_1) \notag \\
&\qquad \times\biggl\{\frac{1}{2m}\left(\frac{\partial E}{\partial x}(\tau(\sigma),\zeta(\sigma)) + \frac{{}^t\partial E}{\partial x}(\tau(\sigma),\zeta(\sigma))\right) - a\biggr\}d\sigma_1d\sigma_2,
\end{align}
\begin{align} \label{5.17}
& E'_1(t,s;x,y,z)  = - \int_0^1\int_0^1 \sigma_1(1-\sigma_1)\notag \\
&\qquad \times\frac{1}{2m}\left(\frac{\partial E}{\partial x}(\tau(\sigma),\zeta(\sigma))  - \frac{{}^t\partial E}{\partial x}(\tau(\sigma),\zeta(\sigma)) \right)d\sigma_1d\sigma_2.
\end{align}
We can now easily prove the following auxiliary result.
\begin{lem}
(1)  Let us define $\Phi(t,s; x,y,z)$ by \eqref{5.10}.  Then we have
\begin{align} \label{5.18}
\frac{\partial \Phi}{\partial z}(t,s;x,y,z) & = I + \rho^2  E'_0(t,s;x,y,z) - \frac{\rho^2a}{6} + \rho^2  E'_1(t,s;x,y,z)\notag \\
 & + B'(t,s;x,y,z),
\end{align}
\begin{equation} \label{5.19}
B'(t,s;x,y,z) = -\rho\frac{\partial}{\partial z}\frac{1}{m}\sum_{k=1}^d(z_k - x_k)\int_0^1\int_0^1\sigma_1B_{jk}(\tau(\sigma),\zeta(\sigma))d\sigma_1d\sigma_2.
\end{equation}
(2)  Let us define $\Phi(t,s; x,y,z)$ by \eqref{5.11}.  Then we have \eqref{5.18} where $B'(t,s;x,y,z)$ is given by
\begin{align} \label{5.20}
& B'(t,s;x,y,z)  = -\rho^2\frac{\partial}{\partial z}\frac{1}{m}\int_0^1d\theta \sum_{k=1}^d(z_k - x_k) \notag \\
& \times \int_0^1\int_0^1\sigma_1(1- \sigma_1)\frac{\partial B_{jk}}{\partial t}(s+\theta(1-\sigma_1)\rho,\zeta(\sigma))d\sigma_1d\sigma_2.
\end{align}
\end{lem}
\begin{lem}
(1)  Let us write   the fifth term on the right-hand side of \eqref{5.10} as $\rho\widetilde{B}(t,s;x,y,z)$.  Assume \eqref{2.10}.  Then we have 
\begin{equation*} 
|\partial_x^{\alpha}\partial_y^{\beta}\partial_z^{\gamma}\widetilde{B}(t,s;x,y,z)| \leq C_{\alpha\beta\gamma}, \ |\alpha + \beta + \gamma| \geq 1.
\end{equation*}
(2)  Let us write  the fifth term on the right-hand side of \eqref{5.11} as $\rho\widetilde{B}(t,s;x,y,z)$.  Assume \eqref{2.12}.  Then we have 
\begin{equation*} 
|\partial_x^{\alpha}\partial_y^{\beta}\partial_z^{\gamma}\widetilde{B}(t,s;x,y,z)| \leq \rho\, C_{\alpha\beta\gamma}, \ |\alpha + \beta + \gamma| \geq 1.
\end{equation*}
\end{lem}
\begin{proof}
Both of (1) and (2) are proved by Lemma 3.5 in \cite{Ichinose 1997}.
\end{proof}
The following lemma is crucial in the present paper.
\begin{lem}
We assume \eqref{2.7} and \eqref{2.8}.  (1)  Let us define $\Phi$ by \eqref{5.11}.  Assume \eqref{2.12} if $0 \leq M_* < 1$ and \eqref{2.13} if $M_* \geq 1$.  Then there exist constants $\rho^* > 0, \delta > 0$ and $C \geq 0$ such that for $0 \leq t - s \leq \rho^*$ and $X = (x,y,z) \in \bR^{3d}$ we have the estimates
\begin{equation} \label{5.21}
\det \frac{\partial \Phi}{\partial z}(\ts,x,y,z) \geq \delta(1 + \rho^2|X|^{2M_*})^d,\ \rho = t -s, 
\end{equation}
\begin{equation} \label{5.22}
\left|\frac{\partial \Phi}{\partial z}(\ts,x,y,z)^{-1} \right|\leq C(1 + \rho^2|X|^{2M_*})^{-1},
\end{equation}
where $|\Omega|$ denotes the Hilbert-Schmidt norm $\bigl( \sum_{i,j=1}^d |\Omega_{ij}|^2\bigr)^{1/2}$ of  a matrix $\Omega =(\Omega_{ij};i\downarrow j \rightarrow 1,2,\dots,d)$ .  
Furthermore, for all fixed $0 \leq t - s \leq \rho^*$ and $(x,y) \in \bR^{2d}$, the map: $\bR^d \ni z \to \xi = \Phi(\ts;x,y,z) \in \bR^d$ is a homeomorphism, whose inverse will be denoted by the map: $\bR^d \ni \xi \to z = z(\ts;x,\xi,y) \in \bR^d$.  (2)  Let us define $\Phi$ by \eqref{5.10}.  Assume \eqref{2.10} and \eqref{2.11}.  Then we have the same assertions as in (1).
\end{lem}
\begin{proof}
We will first prove (1).  Lemma 5.3 and \eqref{5.16} show 
\begin{equation} \label{5.23}
I + \rho^{2}E'_{0}(t,s;x,y,z) \geq 1 + C'_{*}\rho^2|X|^{2M_*}
\end{equation}
for all $X = (x,y,z) \in \bR^{3d}$.  Hence we have
\begin{equation} \label{5.24}
\left|\left(I + \rho^{2}E'_{0}(t,s;x,y,z)\right)^{-1} \right|\leq \frac{C}{1 + \rho^2|X|^{2M_*}}
\end{equation}
together with \eqref{2.8}.  We note Faraday's law
\begin{equation} \label{5.25}
\frac{\partial E}{\partial x}(t,x) - \frac{{}^{t}\partial E}{\partial x}(t,x)  
= -\left(\frac{\partial B_{ji}}{\partial t};i\downarrow j \rightarrow 1,2,\dots,d\right),
\end{equation}
which follows from \eqref{1.1}. Hence, using the assumption \eqref{2.12} if $0\leq M_{*} < 1$ and \eqref{2.13} if $1 \leq M_{*}$, from \eqref{5.17} we have
\begin{equation} \label{5.26}
|\partial_{x}^{\alpha}\partial_{y}^{\beta}\partial_{z}^{\gamma}E'_{1}(t,s;x,y,z)| \leq C_{\alpha\beta\gamma}<X>^{M_{*}}
\end{equation}
for all $\alpha, \beta$ and $\gamma$.  Here we used that if \eqref{2.12} holds, $\partial_x^{\alpha}\partial_tB(t,x)$ are bounded on $\bR^d$ for all $\alpha$.  This follows from Lemma 3.5 in \cite{Ichinose 1997}.
From \eqref{5.20} we also get
\begin{equation} \label{5.27}
|\partial_{x}^{\alpha}\partial_{y}^{\beta}\partial_{z}^{\gamma}B'(t,s;x,y,z)| \leq C_{\alpha\beta\gamma}\rho^{2}<X>^{M_{*}}
\end{equation}
for all $\alpha, \beta$ and $\gamma$ together with (2) of Lemma 5.5.   
\par
	Noting \eqref{5.23}, we can rewrite \eqref{5.18} as 
\begin{align} \label{5.28}
& \frac{\partial \Phi}{\partial z}(t,s;x,y,z)  = \left(I + \rho^2  E'_0\right) \Big\{ I -\left(I + \rho^2  E'_0\right)^{-1}\frac{\rho^2a}{6} +  \rho\left(I + \rho^2  E'_0\right)^{-1} \notag \\ 
&\times \rho E'_1(t,s;x,y,z) + \rho\left(I + \rho^2  E'_0\right)^{-1}\rho ^{-1}B'(t,s;x,y,z)\Bigr\}.
\end{align}
Noting $\theta/(1+\theta^{2}) \leq 1$ for all $\theta \geq 0$, from \eqref{2.8}, \eqref{5.24} and \eqref{5.26} we have
\begin{equation} \label{5.29}
|\partial_{x}^{\alpha}\partial_{y}^{\beta}\partial_{z}^{\gamma} \left(I + \rho^2  E'_0\right)^{-1} \rho E'_{1}(t,s;x,y,z)| \leq C_{\alpha\beta\gamma} < \infty
\end{equation}
 for all $\alpha, \beta$ and $\gamma$.  In the same way from \eqref{5.27} we also have
\begin{equation} \label{5.30}
|\partial_{x}^{\alpha}\partial_{y}^{\beta}\partial_{z}^{\gamma} \left(I + \rho^2  E'_0\right)^{-1}\rho^{-1} B'(t,s;x,y,z)| \leq C_{\alpha\beta\gamma} < \infty
\end{equation}
for all $\alpha, \beta$ and $\gamma$.  Therefore, from \eqref{5.23} and \eqref{5.28} we have
\begin{equation*} 
\det \frac{\partial \Phi}{\partial z} \geq \bigl( I + C'_{*}\rho^{2}|X|^{2M_{*}}\bigr)^{d}(1 - C\rho)
\end{equation*}
with a constant $C \geq 0$.  Thereby we can see together with \eqref{5.24} and \eqref{5.28}-\eqref{5.30} that there exists a constant $\rho^{*} > 0$ satisfying \eqref{5.21} and \eqref{5.22}.  Hence we can complete the proof of the assertion (1) by using Theorem 1.22 on p. 16 in \cite{Schwartz}.
\par
	We will prove (2).  As in the proof of (1) we can prove \eqref{5.23} - \eqref{5.26}, and so prove \eqref{5.29}.  Now, $B'$ is given by \eqref{5.19}.  Then from (1) of Lemma 5.5 we have
\begin{equation} \label{5.31}
|\partial_{x}^{\alpha}\partial_{y}^{\beta}\partial_{z}^{\gamma}B'(t,s;x,y,z)| \leq C_{\alpha\beta\gamma}\rho
\end{equation}
 for all $\alpha, \beta$ and $\gamma$.  Consequently we can prove \eqref{5.30}.  Hence we can complete the proof of (2) as in the proof of (1).
\end{proof}
The constant $\rho^* > 0$ defined in Lemma 5.6 is fixed from now on throughout sections 5, 6 and 7.
\begin{pro}
We assume \eqref{2.7} - \eqref{2.9}.  (1)  Let us define $\Phi$ by \eqref{5.11}.  Assume \eqref{2.12} if $0 \leq M_* < 1$ and \eqref{2.13} if $M_* \geq 1$.  Let  $0 \leq t - s \leq \rho^*$ and  $z(\ts; x,\xi,y)$  the function defined in Lemma 5.6.  Then we have
\begin{equation}   \label{5.32}
            |\partial_{\xi}^{\alpha}\partial_x^{\beta}\partial_{y}^{\gamma}
            z_j(t,s;x,\xi,y)|  \leq C_{\alpha,\beta,\gamma},
\  |\alpha + \beta  + \gamma| \geq 1
           \end{equation}
for  $(x,\xi,y) \in \bR^{3d}.$  (2) Let us define $\Phi$ by \eqref{5.10}.  Assume \eqref{2.10} and \eqref{2.11}.  Then we have the same assertions as in (1).
\end{pro}
\begin{proof}
Let $0 \leq t - s \leq \rho^*$.  We will first prove (1).  Let $w = x, \xi$ or $y$. It follows from $\xi = \Phi(\ts;x,y,z(t,s;x,\xi,y))$ that we have
\begin{equation}   \label{5.33}
 \frac{\partial \xi}{\partial w_j} =  \frac{\partial \Phi}{\partial z}(t,s;x,y,z) \frac{\partial z}{\partial w_j} +  \frac{\partial \Phi}{\partial w_j}(t,s;x,y,z)
           \end{equation}
           and so from (1) of Lemma 5.6
\begin{equation*}   
 \frac{\partial z}{\partial w_j}(\ts;x,\xi,y) =  \left(\frac{\partial \Phi}{\partial z}\right)^{-1}\left( \frac{\partial \xi}{\partial w_j} - \frac{\partial \Phi}{\partial w_j} \right).
           \end{equation*}
  Using \eqref{2.8}-\eqref{2.9} and \eqref{2.12}-\eqref{2.13}, from \eqref{5.11} and (2) of Lemma 5.5 we get
\begin{align}   \label{5.34}
&\left| \frac{\partial \Phi}{\partial w_j}(\ts;x,y,z)\right| \leq  C( 1 + \rho|X|^{M_*} + \rho^2|X|^{2M_*} + \rho^2|X|^{M_*}) \notag \\
& \leq C'( 1 + \rho|X|^{M_*} + \rho^2|X|^{2M_*})
           \end{align}
with non-negative constants $C$ and $C'$.  Hence, using \eqref{5.22}, we can prove
\begin{equation*}   
\left| \frac{\partial z}{\partial w_j}(\ts;x,\xi,y)\right| \leq C'' < \infty
           \end{equation*}
with a constant $C'' \geq 0$.  Next from \eqref{5.33} we have
\begin{align*}   
0 & =  \frac{\partial \Phi}{\partial z} \frac{\partial^2 z}{\partial w_k\partial w_j} +          
\left(\frac{\partial^2 \Phi}{\partial w_k\partial z} + \frac{\partial^2 \Phi}{\partial z^2} \frac{\partial z}{\partial w_k}\right) \frac{\partial z}{\partial w_j} \notag \\
& +  \frac{\partial^2 \Phi}{\partial w_k\partial w_j} + \frac{\partial^2 \Phi}{\partial z\partial w_j} \frac{\partial z}{\partial w_k}.
   \end{align*}
Hence,  using \eqref{5.32} with $|\alpha + \beta + \gamma| = 1,$ we can prove \eqref{5.32} with $|\alpha + \beta + \gamma| = 2$ as in the proof of the case of  $|\alpha + \beta + \gamma| = 1$.  In the same way we can complete the proof of \eqref{5.32} by induction.
\par
	 We consider the assertion (2). $\Phi$ is given by \eqref{5.10}.  Then we see from (1) of Lemma 5.5 that the corresponding inequalities to \eqref{5.34}   are given by 
\begin{equation*}   
\left| \frac{\partial \Phi}{\partial w_j}(\ts;x,y,z)\right|  \leq C( 1 + \rho|X|^{M_*} + \rho^2|X|^{2M_*} + \rho)
           \end{equation*}
           Hence we can prove \eqref{5.32} as in the proof of (1).
\end{proof}
\begin{thm}
Suppose \eqref{2.3} and Assumption 2.2.  Let $\Cts(t,s)$ be the operator on $\Sspace(\bR^d)$ defined by \eqref{5.6} and $0 \leq t - s \leq \rho^*$.  Then $\Cts(\ts)$ can be extended to a bounded operator on $L^2(\bR^d)$ and satisfies 
\begin{equation}   \label{5.35}
 \Vert \Cts(\ts)f\Vert \leq e^{K(t-s)}\Vert f \Vert
           \end{equation}
           for all $f \in L^2$ with a constant $K \geq 0$.
\end{thm}
\begin{proof}
Since we can prove from \eqref{2.3} and \eqref{2.8} as in the proof of \eqref{2.23.3} that  $\partial_x^{\alpha}\partial_tA_j(t,x)$ are continuous in $\domain$, Lemma 5.2 holds.
We will first prove the case that \eqref{2.12} and \eqref{2.13} are assumed.  Let us define $\Phi$ by \eqref{5.11}.  Then from Lemma 5.2 we have
\begin{align*}  
    & \Cts(\ts)^{\dag}\chi(\epsilon\cdot)^2\Cts(\ts)f = \left(\frac{m}{2\pi(t-s)}\right)^d\int f(y)dy \\
    &\quad \times \int\chi(\epsilon z)^2
      \left(\exp i(x - y)\cdot \frac{m\Phi}{t-s} \right)dz
\end{align*}
and so, changing variables from $z$ to $\xi = \Phi(\ts;x,y,z)$ by Lemma 5.6, 
\begin{align}  \label{5.36}
    & \Cts(\ts)^{\dag}\chi(\epsilon\cdot)^2\Cts(\ts)f = \left(\frac{m}{2\pi(t-s)}\right)^d\int f(y)dy \int\chi(\epsilon z)^2 \notag \\
    &\quad \times 
      \left(\exp i(x - y)\cdot \frac{m\xi}{t-s} \right)\det \frac{\partial z}{\partial \xi}\,d\xi = \iint e^{i(x-y)\cdot\eta}\chi(\epsilon z(\ts;x,\xi,y))^2 \notag \\
      & \quad \times \det\frac{\partial z}{\partial \xi}(\ts;x,\xi,y)f(y)dy\dbar\eta,\ \xi =\frac{t-s}{m}\eta.
\end{align}
From \eqref{2.8}, \eqref{5.24} and \eqref{5.28} - \eqref{5.30} we have
\begin{align}  \label{5.37}
    & 0 < \det\frac{\partial z}{\partial\xi}(\ts;x,\xi,y) = \det\bigl(I +\rho^2E'_0\bigr)^{-1} + (t - s)p_1(\ts;x,\xi,y) 
    \notag \\
    & \equiv p_0(\ts;x,\xi,y) + (t-s)p_1(\ts;x,\xi,y) 
      \end{align}
with $p_j(\ts;x,\xi,y) \in S^0(\bR^{3d})\ (j = 0,1)$.  In particular, from \eqref{5.23} we have
\begin{equation}  \label{5.38}
    0 \leq p_0(\ts;x,\xi,y) \leq 1.
      \end{equation}
Noting \eqref{5.32},  from \eqref{5.36} and \eqref{5.37} we can prove
\begin{align*}  
    &\limepsilon \Cts(\ts)^{\dag}\chi(\epsilon\cdot)^2\Cts(\ts)f = \iint e^{i(x-y)\cdot\eta}p_0(\ts;x,\frac{t-s}{m}\eta,y)f(y)dy\dbar\eta \\
    & \quad + (t - s)\iint e^{i(x-y)\cdot\eta}p_1(\ts;x,\frac{t-s}{m}\eta,y)f(y)dy\dbar\eta 
    \end{align*}
in $\Sspace$ for $f \in \Sspace$.  Therefore, applying Theorem 1.A to the above, we have
\begin{align*}  
    & \Vert\Cts(\ts)f\Vert^2 \leq \liminf_{\,\epsilon\to 0+}\left(\Cts(\ts)^{\dag}\chi(\epsilon\cdot)^2\Cts(\ts)f,f\right) \leq \bigl(1 + K(t-s)\bigr)\Vert f\Vert^2
    \\
    & 
    \quad + K(t-s)\Vert f\Vert^2 = \bigl(1 + 2K(t-s)\bigr)\Vert f\Vert^2 \leq e^{2K(t-s)}\Vert f\Vert^2
\end{align*}
with a constant $K \geq 0$, which shows \eqref{5.35}.
\par
	Next we consider the case that \eqref{2.10} and \eqref{2.11} are assumed.  Let us define $\Phi$ by \eqref{5.10}.  Then we can prove Theorem 5.8 as in the proof of the first case.
\end{proof}
\begin{cor}
Suppose the assumptions of Theorem 5.8.  Let $\kdelta(t,0)f$ for $f \in \Cspace_0(\bR^d)$ be the approximation defined by \eqref{2.1} of the Feynman path integral.  Let $|\Delta| \leq \rho^*$.  Then $\kdelta(t,0)f$ can be uniquely extended to a bounded operator on $L^2(\bR^d)$, which can be written as 
\begin{equation}  \label{5.39}
     \kdelta(t,0)f = {\cal C}(t,\tau_{\nu-1}){\cal C}(\tau_{\nu-1},\tau_{\nu-2})\cdots 
{\cal C}(\tau_1,0)f
\end{equation}
for $f \in L^2$, and one has \eqref{2.14} with the same constant $K$ as in Theorem 5.8.
\end{cor}
\begin{proof}
As in the proof of \eqref{4.11} from Theorem 5.8 we can prove \eqref{5.39},
which shows \eqref{2.14} by \eqref{5.35}.
\end{proof}
\begin{thm}
Suppose \eqref{2.3} and Assumption 2.2.  Let $q(x,w)$ be a function satisfying \eqref{5.3} with $M = 0$.  We set $p(\ts;x,w) = q(x,\sqrt{t-s}w)$ and define $P(\ts)$ by \eqref{5.4}.  Then we have
\begin{equation}  \label{5.40}
    \Vert P(\ts)f\Vert \leq C\Vert f\Vert, \ 0 \leq t - s \leq \rho^*
\end{equation}
for $f \in L^2$ with a constant $C \geq 0$.
\end{thm}
\begin{proof}
Let $f \in \Sspace$.  As in the proof of \eqref{5.36} we have
\begin{align*}  
    & P(\ts)^{\dag}\chi(\epsilon\cdot)^2P(\ts)f = \left(\frac{m}{2\pi(t-s)}\right)^d\int f(y)dy \int\chi(\epsilon z)^2
     \left(\exp i(x - y)\cdot \frac{m\xi}{t-s} \right)  \notag \\
    &\quad \times 
     q(z,z-x)^*q(z,z-y)\det \frac{\partial z}{\partial \xi}(\ts;x,\xi,y)\,d\xi. 
\end{align*}
Noting \eqref{5.32}, we have 
\begin{align}  \label{5.41}
    & \limepsilon P(\ts)^{\dag}\chi(\epsilon\cdot)^2P(\ts)f = \iint e^{i(x - y)\cdot \eta}
     q(z,z-x)^*q(z,z-y) \notag \\
     &\quad \times \det \frac{\partial z}{\partial \xi}(\ts;x,\frac{t-s}{m}\eta,y)f(y)dy\,\dbar\eta
\end{align}
in $\Sspace$ with $z = z(\ts;x,(t-s)\eta/m,y)$  and hence we can prove \eqref{5.40} by Theorem 1.A as in the proof of \eqref{5.35}.
\end{proof}
\section{Consistency of $\Cts(\ts)$}
\begin{lem}
Suppose the assumptions of Proposition 5.7.  Let $0 \leq t-s \leq \rho^*$ and $z(\ts;x,\xi,y)$ the function defined  in Lemma 5.6.  Then we have
\begin{align}  \label{6.1}
    & \left| z(\ts;x,\rho\eta/m + \sqrt{\rho}\zeta/m,x+\sqrt{\rho}y) - x\right| \leq C\sqrt{\rho}(1 + \sqrt{\rho}|x|^{2M_*+1} \notag \\
     & \quad  +  |y|^{2M_*+1} + \sqrt{\rho}|\eta|^{2M_*+1} + |\zeta|^{2M_*+1}).
     \end{align}
\end{lem}
\begin{proof}
We first consider the case that $\Phi$ is given by \eqref{5.11}, which we write as
\begin{align}  \label{6.2}
        \Phi(\ts;x,y,z)& =  z - \frac{x + y}{2} + \rho\widetilde{A}(s;x,y)    + \rho^2 \widetilde{E}(\ts;x,y,z)   & 
       \notag \\
       &  + \rho \widetilde{B}(\ts;x,y,z).
\end{align}
Then, using (2) of Lemma 5.5, from the assumptions \eqref{2.8} - \eqref{2.9} and \eqref{2.12} - \eqref{2.13} we have
\begin{equation}   \label{6.3}
            |\partial_{x}^{\alpha}\partial_y^{\beta}
           \widetilde{A}(s;x,y)|  \leq C_{\alpha,\beta}(1 + |x| + |y|)^{M_*}, \ |\alpha + \beta| \geq 1,
           \end{equation}
\begin{equation}   \label{6.4}
            |\partial_{x}^{\alpha}\partial_y^{\beta}\partial_z^{\gamma}
           \widetilde{E}_j(\ts;x,y,z)|  \leq C_{\alpha,\beta,\gamma}(1 + |x| + |y|+ |z|)^{2M_*}, \  |\alpha + \beta| + \gamma| \geq 1,
           \end{equation}
\begin{equation}   \label{6.5}
            |\partial_{x}^{\alpha}\partial_y^{\beta}\partial_z^{\gamma}
           \widetilde{B}_j(\ts;x,y,z)|  \leq C_{\alpha,\beta,\gamma}\rho(1 + |x| + |y|+ |z|)^{M_*}, \  |\alpha + \beta| + \gamma| \geq 1.
           \end{equation}
From \eqref{5.32} we have
\begin{equation*}   
         |z(\ts;x,\xi,y)| \leq C(1 + |x| + |\xi|+ |y|),
\end{equation*}
which shows 
\begin{align}  \label{6.6}
    & \left| z(\ts;x,\rho\eta/m + \sqrt{\rho}\zeta/m,x+\sqrt{\rho}y)\right|  \notag \\
    & \leq C(1 + |x| + \sqrt{\rho}|y| + \rho|\eta| + \sqrt{\rho}|\zeta|).
     \end{align}
We take $z =z(\ts;x,\rho\eta/m + \sqrt{\rho}\zeta/m,x+\sqrt{\rho}y)$ in \eqref{6.2}.  Then we have
\begin{align*}  
      &  \frac{\rho\eta}{m} +  \frac{\sqrt{\rho}\zeta}{m}  =  z - \frac{2x +\sqrt{\rho} y}{2} + \rho\widetilde{A}(s;x,x +\sqrt{\rho} y)   \\& 
       \quad  \quad + \rho^2 \widetilde{E}(\ts;x,x +\sqrt{\rho} y,z)  + \rho \widetilde{B}(\ts;x,x +\sqrt{\rho} y,z).
\end{align*}
Hence we get
\begin{align}  \label{6.7}
      &  \frac{z-x}{\sqrt{\rho}} =  \frac{1}{2}\,y + \frac{\sqrho\eta}{m} +  \frac{\zeta}{m}  - \sqrho\widetilde{A}(s;x,x +\sqrt{\rho} y)   \notag \\ 
       & \quad   - \rho^{3/2} \widetilde{E}(\ts;x,x +\sqrt{\rho} y,z)  - \sqrho \widetilde{B}(\ts;x,x +\sqrt{\rho} y,z).
\end{align}
Applying \eqref{6.3} - \eqref{6.6} to \eqref{6.7}, we have \eqref{6.1}.
\par
	We consider the case that $\Phi$ is given by \eqref{5.10}.  We write $\Phi$ as \eqref{6.2}. Then from (1) of Lemma 5.5 we have
\begin{equation}   \label{6.8}
            |\partial_{x}^{\alpha}\partial_y^{\beta}\partial_z^{\gamma}
           \widetilde{B}_j(\ts;x,y,z)|  \leq C_{\alpha,\beta,\gamma}, \  |\alpha + \beta + \gamma| \geq 1
\end{equation}
correspondingly to \eqref{6.5}.  Hence we can also prove \eqref{6.1}.
\end{proof}
From \eqref{5.32} we have the following.
\begin{lem}
Suppose the assumptions of Proposition 5.7.  Let $0 \leq t-s \leq \rho^*$.  Then we have
\begin{align}  \label{6.9}
    & \left| \partial_{\eta}^{\alpha}\partial_{y}^{\beta}\partial_{\zeta}^{\gamma}\bigl(z(\ts;x,\rho\eta/m + \sqrt{\rho}\zeta/m,x+\sqrt{\rho}y)\bigr) \right|    \leq C_{\alpha\beta\gamma}\sqrt{\rho}^{\, |2\alpha + \beta + \gamma|},\notag \\
    & \hspace{1cm} |\alpha + \beta + \gamma| \geq 1.
          \end{align}
\end{lem}
\begin{lem}
Suppose the assumptions of Proposition 5.7.  Let $0 \leq t - s \leq \rho^*.$  Take a $p(x,w)$ satisfying \eqref{5.3} and set 
\begin{align}  \label{6.14}
    & q_{\epsilon}(\ts;x,\eta) = \text{Os} - \iint    e^{-iy\cdot\zeta}p\left(z,\frac{z-x}{\sqrho}\right)^*\chi(\epsilon z)^2
    p\left(z,\frac{z-x-\sqrho y}{\sqrho}\right) \notag \\
    & \times \det\frac{\partial z}{\partial \xi}(\ts;x,\rho\eta/m + \sqrt{\rho}\zeta/m,x+\sqrt{\rho}y)dy\,\dbar\zeta
    \end{align}
for $0 < \epsilon \leq 1$, where $z = z(\ts;x,\rho\eta/m + \sqrt{\rho}\zeta/m,x+\sqrt{\rho}y)$.  Then we have
\begin{equation}   \label{6.15}
      |\partial_{\eta}^{\alpha}q_{\epsilon}(\ts;x,\eta)| \leq C_{\alpha}(1 + |x|+ |\eta|)^{2M(2M_* + 1)}
\end{equation}
for all $\alpha$ with constants $C_{\alpha}$ independent of $\epsilon$.
\end{lem}
\begin{proof}
We write $<D_y>^2\, = 1 - \sum_{j=1}^d\partial_{y_j}^2$. Let $l_j \geq 0\ (j = 0,1)$ be integers.  Using \eqref{6.1}, \eqref{6.6} and \eqref{6.9}, from \eqref{6.14} we have
\begin{align*}  
    & |q_{\epsilon}(\ts;x,\eta)| \leq  \ \iint    \biggl| <y>^{-2l_0}<D_{\zeta}>^{2l_0}<\zeta>^{-2l_1}<D_y>^{2l_1}\biggl\{p\left(z,\frac{z-x}{\sqrho}\right)^* \chi(\epsilon z)^2\\
    & \quad 
   \times  p\left(z,\frac{z-x-\sqrho y}{\sqrho}\right) 
    \det\frac{\partial z}{\partial \xi}\biggr\}\biggr|dy\,\dbar\zeta \leq C_1\iint <y>^{-2l_0}<\zeta>^{-2l_1}\\
    & \times\Bigl(1 + |z| + \left|\frac{z-x}{\sqrho}\right| + |y|\Bigr)^{2M}dy\,\dbar\zeta   \leq C_2\iint <y>^{-2l_0}<\zeta>^{-2l_1}\\ 
    & \times\Bigl(1 + |x|^{2M_* +1}+ |y|^{2M_* +1}+ |\eta|^{2M_* +1} + |\zeta|^{2M_* +1}\Bigr)^{2M}dy\,\dbar\zeta
  \leq C_3\iint <y>^{-2l_0} \\
 & \times <\zeta>^{-2l_1}<y>^{2M(M_*+1)}<\zeta>^{2M(M_*+1)}\Bigl(1+|x|^{2M_*+1} + |\eta|^{2M_*+1}\Bigr)^{2M}dy\,\dbar\zeta.
    \end{align*}
    Hence, taking $l_0$ and $l_1$ so that $2l_j - 2M(2M_*+1) > d,$  we get 
\begin{equation*}  
      |q_{\epsilon}(\ts;x,\eta)| \leq C_{4}(1 + |x|+ |\eta|)^{2M(2M_* + 1)}.
\end{equation*}
In the same way we can prove \eqref{6.15}, using \eqref{6.9}.
\end{proof}
\begin{pro}
Suppose \eqref{2.3} and Assumption 2.2.  Let $p(x,w)$ be a function satisfying \eqref{5.3} and define $P(\ts)$ by \eqref{5.4}.  Let $0 \leq t-s \leq \rho^*$.  Then there exists an integer $l \geq 0$ such that we have
\begin{equation}  \label{6.16}
     \Vert P(\ts)f\Vert \leq C |f|_l
\end{equation}
for $f \in \Sspace(\bR^d)$.
\end{pro}
\begin{proof}
If $t = s$, the inequality \eqref{6.16} follows from \eqref{5.4}.  Let $0 < t - s \leq \rho^*$.
Let us define $q_{\epsilon}(\ts;x,\eta)$ by \eqref{6.14} for $p(x,w)$.  Then, using Lemma 5.2 and \eqref{5.32}, we can prove
\begin{equation}  \label{6.17}
    P(\ts)^{\dag}\chi(\epsilon\cdot)^2P(\ts)f = Q_{\epsilon}(\ts;X,D_x)f
\end{equation}
for $f \in \Sspace$, which has been proved  in (4.12) of \cite{Ichinose 1999}.  Hence we see
\begin{align*}  
    & \Vert \chi(\epsilon\cdot)P(\ts)f\Vert^2    = (P(\ts)^{\dag}\chi(\epsilon\cdot)^2P(\ts)f,f) = (Q_{\epsilon}(t,s)f,f) \\
    & = \int f(x)^*dx \int e^{i(x-y)\cdot\eta}q_{\epsilon}(\ts;x,\eta)f(y)dy\,\dbar\eta \\
    & = \int f(x)^*dx \int e^{i(x-y)\cdot\eta}<\eta>^{-2l_0}<D_y>^{2l_0}q_{\epsilon}(\ts;x,\eta)f(y)dy\,\dbar\eta.
    \end{align*}
    Consequently, using \eqref{6.15}, we have
\begin{align}  \label{6.18}
    & \Vert P(\ts)f\Vert^2   \leq C_1 \int |f(x)|dx \int <\eta>^{-2l_0}(1 + |x| + |\eta|)^{2M(2M_*+1)} \notag \\
    &\quad \times |<D_y>^{2l_0}f(y)|dy\,\dbar\eta \leq C_1\int <x>^{2M(2M_*+1)}|f(x)|dx \notag \\
    & \times \int <\eta>^{-2l_0+2M(2M_*+1)}\dbar\eta \int  |<D_y>^{2l_0}f(y)|dy \notag \\
    & \leq C_2 \int <\eta>^{-2l_0+2M(2M_*+1)}\dbar\eta |f|_l^2
    \end{align}
    with an integer $l$.  Taking $l_0$ so that $2l_0 - 2M(2M_*+1) > d$, we obtain \eqref{6.16}.
\end{proof}
\begin{thm}
Suppose \eqref{2.3} and Assumption 2.2.  Let $p(x,w)$ be a function satisfying \eqref{5.3} and define $P(\ts)$ by \eqref{5.4}.  Let $0 \leq t-s \leq \rho^*$.   Then, for any $\alpha$ there exists an integer $l(\alpha) \geq 0$ such that we have
\begin{equation}  \label{6.19}
      \Vert x^{\alpha}\bigl(P(\ts)f\bigr)\Vert \leq C |f|_{l(\alpha)},\quad
     \Vert \partial_x^{\alpha}\bigl(P(\ts)f\bigr)\Vert \leq C |f|_{l(\alpha)}
\end{equation}
for $f  \in \Sspace$.
\end{thm}
\begin{proof}
Setting $p'(x,s) = x^{\alpha}\,p(x,w)$, we have $P'(\ts)f = x^{\alpha}P(\ts)f$ from \eqref{5.4}.  Hence from Proposition 6.4 we can prove
\begin{equation*}  
      \Vert x^{\alpha}\bigl(P(\ts)f\bigr)\Vert = \Vert P'(\ts)f\Vert\leq C |f|_{l(\alpha)}
     \end{equation*}
for $f  \in \Sspace$ with an $l(\alpha) \geq 0$, which shows the first inequality of \eqref{6.19}.
\par
Next we can write $P(\ts)f$ as 
\begin{align}  \label{6.20}
      & P(t,s)f = \sqrt{\frac{m}{2\pi i}}^{\ d}
       \text{Os}- \int e^{i\phi(t,s;x,w)}p(x,w)f(x - \sqrt{\rho}w)dw,    \notag \\
    & \phi(t,s;x,w) = \frac{m}{2}|w|^2 + \sqrt{\rho}w\cdot
\int^1_0 A(t-\theta\rho, x - \theta\sqrt{\rho}w)d\theta    \notag    \\
& \qquad
  - \rho\int^1_0 V(t-\theta\rho, x - \theta\sqrt{\rho}w)d\theta , \quad 
\rho = t - s
\end{align}
as in \S 2 of \cite{Ichinose 1999}.  Then we have
\begin{equation} \label{6.21}
\partial_x^{\alpha}\bigl(P(\ts)f\bigr) = \sum_{\beta \leq \alpha}P_{\beta}(\ts)(\partial_x^{\alpha - \beta}f),
\end{equation}
where $\beta \leq \alpha$ indicates $\beta_j \leq \alpha_j$ for all $j = 1,2,\dots,d$.  Using the assumptions \eqref{2.3} and \eqref{2.9}, from \eqref{6.20} we have
\begin{equation} \label{6.22}
|\partial_w^{\alpha'}\partial_x^{\beta'}p_{\beta}(\ts;x,w)| \leq C_{\alpha'\beta'}(1 + |x| + |w|)^{M+2(M_*+1)|\beta|}
\end{equation}
for all $\alpha'$ and $\beta'$.  Hence, applying Proposition 6.4 to $P_{\beta}(\ts)$, from \eqref{6.21} we obtain the second inequality of \eqref{6.19}.
\end{proof}
\begin{pro}
We assume \eqref{2.3}, \eqref{2.8} and \eqref{2.9}.  Let $H(t)$ and $\Cts(\ts)$ be the operators defined by \eqref{1.3} 
and \eqref{5.6}, respectively.  Then there exists a continuous function $r(\ts;x,w)$ in $0 \leq s \leq t \leq T$ and $(x,w) \in \bR^{2d}$ satisfying \eqref{5.3} for an $M \geq 0$ such that
\begin{equation}  \label{6.23}
      \left\{i\frac{\partial}{\partial t} - H(t)\right\}\Cts(\ts) f
                        = \sqrt{t - s} R(\ts)f
\end{equation}
for $f \in \Sspace(\bR^d)$.
\end{pro}
\begin{proof}
From \eqref{2.3} and \eqref{2.8} we have
\begin{equation} \label{6.24}
     | \partial_x^{\alpha} \partial_tA(t,x)| \leq C_{\alpha}<x>^{2(M_*+1)}
\end{equation}
for all $\alpha$ as in the proof of \eqref{2.23.3}.  Consequently we get Proposition 6.6 from Proposition 3.5 in \cite{Ichinose 2007} or Proposition 2.3 in \cite{Ichinose 1997}.
\end{proof}
\section{Proofs of Theorems 2.1 and 2.2}
We suppose Assumption 2.1 and let $M_* \geq 0$ be the constant in Assumption 2.1.  Let us introduce the weighted Sobolev spaces
\begin{align} \label{7.0}
B^a(\mathbb{R}^d) & := \{f \in  L^2(\mathbb{R}^d);
 \|f\|_a := \|f\| + \sum_{|\alpha| \leq  2a} \|\partial_x^{\alpha}f\| \notag\\
 &  +
\|<\cdot>^{2a(M_*+1)}f\| < \infty\}\ (a = 1,2,\dots).
 \end{align}
 We denote the dual space of $B^a$ by $B^{-a}$ and the $L^2$ space by $B^0$.  
\par
We have proved the following in Theorem 2.1 of \cite{Ichinose 2018} and its proof. 
\par
\vspace{0.5cm}
 T{\sc heorem} 7.A.  {\it Suppose Assumption 2.1 and \eqref{2.9}.  Then for any $f \in B^a\ (a = 0, \pm1,\pm2,\dots)$
 there exists a solution $u(t) = U(t,0)f \in C^0_t([0,T];B^a) \cap C^1_t([0,T];B^{a-1})$  with $u(0) = f$ to the equation \eqref{1.3}. This solution $u(t)$ is uniquely determined in the space $\bigcup_{a'=1}^{\infty}C^0_t([0,T];B^{-a'}) \cap C^1_t([0,T];B^{-a'-1})$.  We also have
 \begin{equation} \label{7.1}
 \Vert u(t)\Vert_a \leq C_a  \Vert f\Vert_a, \ 0 \leq t \leq T
 \end{equation}
 and in particular
 \begin{equation} \label{7.2}
 \Vert u(t)\Vert =  \Vert f\Vert,  \ 0 \leq t \leq T.
 \end{equation}
 }
 %
 C{\sc orollary} 7.B.  {\it Suppose Assumption 2.1 and \eqref{2.9}. 
 Then for any integer $l \geq 0$ there exists an integer $l' \geq 0$ such that 
 \begin{equation} \label{7.3}
 | U(t,0)f|_l \leq C_l |f|_{l'},  \ 0 \leq t \leq T
 \end{equation}
 for all $f \in \Sspace$.
 }
\begin{proof}
The Sobolev lemma indicates 
 \begin{equation*} 
\sup_{x\in \bR^d}|f(x)| \leq C\sum_{|\alpha| \leq[d/2] +1}\Vert \partial_x^{\alpha}f\Vert,
 \end{equation*}
 where $[\cdot]$ denotes the Gauss symbol (cf. (2.24) on p. 78 in \cite{Mizohata}).  Hence, for any integer $l \geq 0$ there exist integers $l_1 \geq 0$ and $l_2 \geq 0$ such that 
 \begin{equation} \label{7.4}
 | f|_l \leq C \Vert f\Vert_{l_1},  \  \Vert f\Vert_{l} \leq C' |f|_{l_2} 
 \end{equation}
  for  $f \in \Sspace$.  Therefore from \eqref{7.1} we have
 \begin{equation*} 
 | U(t,0)f|_l \leq C \Vert U(t,0)f\Vert_{l_1} \leq CC_{l_1}\Vert f\Vert_{l_1} \leq C'_{l_1}|f|_{l'}
 \end{equation*}
with an integer $l'$.
\end{proof}
\begin{lem}
Suppose \eqref{2.3} - \eqref{2.4} and Assumption 2.2.  Let $H(t)$ and $\Cts(\ts)$ be the operators defined by \eqref{1.3} and \eqref{5.6}, respectively.  Then there exists an integer $l \geq 0$ such that 
 \begin{equation} \label{7.5}
\left\Vert \frac{\Cts(\ts)f - f}{t-s} - H(t)f\right\Vert \leq C\sqrho|f|_l, \ 0 <  t - s \leq \rho^*
 \end{equation}
 for all $f \in \Sspace$.
\end{lem}
\begin{proof}
Using \eqref{6.23}, we can write
 \begin{align} \label{7.6}
&i\bigl\{ \Cts(\ts)f - f\bigr\} = i\bigl\{ \Cts(s+\rho,s)f - f\bigr\} = i\rho\int_0^1\frac{\partial\Cts}{\partial t}(s+\theta\rho,s)fd\theta \notag \\
& = \rho\int_0^1\bigl\{H(s+\theta\rho)\Cts(s+\theta\rho,s)f + \sqrho R(s+\theta\rho,s)f \bigr\}d\theta
\end{align}
and so
 \begin{align} \label{7.7}
& i\frac{\Cts(\ts)f - f}{\rho} - H(t)f = \sqrho\int_0^1R(s+\theta\rho,s)fd\theta + \int_0^1H(s+\theta\rho)
\notag \\
& \cdot \bigl\{\Cts(s+\theta\rho,s)f - f\bigr\}d\theta + \int_0^1\bigl\{H(s+\theta\rho) - H(t)\bigr\}fd\theta.
\end{align}
From \eqref{2.3} and \eqref{2.9} we can see that for $a = 0,1,2,\dots$  there exist integers $l(a) \geq 0$ satisfying 
\begin{equation*}
\Vert H(t)f \Vert_a \leq C_a \Vert f\Vert_{l(a)}.
\end{equation*}
Consequently we see that the $L^2$ norm of the second term on the right-hand side of \eqref{7.7} is bounded by
\[
C \int_0^1\Vert \mathcal{C}(s + \theta\rho,s)f - f\Vert_{l(0)}d\theta.
\]
Applying \eqref{7.6} to this term, and applying \eqref{2.4}, \eqref{2.9} and \eqref{6.24} to the third term on the right-hand side of \eqref{7.7}, we have
 \begin{align} \label{7.8}
& \left\Vert i\frac{\Cts(\ts)f - f}{\rho} - H(t)f \right\Vert\leq  \sqrho\int_0^1\Vert R(s+\theta\rho,s)f\Vert d\theta + C_1\int_0^1d\theta \int_0^1\rho   \notag \\
& \times \bigl\{\Vert \Cts(s+\theta'\theta\rho,s)f\Vert_{l'(0)}+ \sqrho\Vert R(s+\theta'\theta\rho,s)f\Vert_{l(0)}\bigr\}d\theta' + C_2\rho\Vert f\Vert_{l'(0)}
\end{align}
with an integer $l'(0) \geq 0$.
Hence, applying Theorem 6.5 to $\Cts(\ts)f$ and $R(\ts)f$, and using \eqref{7.4}, we can prove \eqref{7.5}.
\end{proof}
\begin{lem}
Suppose Assumptions 2.1 and 2.2.  Then there exists an integer $l \geq 0$ such that we have
 \begin{equation} \label{7.9}
\left\Vert \Cts(\ts)f  - U(\ts)f\right\Vert \leq C\sqrho\,^{3}|f|_l, \ 0 <  t - s \leq \rho^*
 \end{equation}
 for $f \in \Sspace$.
\end{lem}
\begin{proof}
Correspondingly to \eqref{7.6} - \eqref{7.8} we have
 \begin{equation*} 
i\bigl\{ U(\ts)f - f\bigr\} = \rho \int_0^1H(s+\theta\rho)U(s+\theta\rho,s)f d\theta,
\end{equation*}
 \begin{align*} 
& i\frac{U(\ts)f - f}{\rho} - H(t)f =  \int_0^1H(s+\theta\rho) \bigl\{U(s+\theta\rho,s)f - f\bigr\}d\theta
\notag \\
&\quad  + \int_0^1\bigl\{H(s+\theta\rho) - H(t)\bigr\}fd\theta,
\end{align*}
 \begin{equation*} 
 \left\Vert i\frac{U(\ts)f - f}{\rho} - H(t)f \right\Vert\leq  C_1\int_0^1d\theta \int_0^1\rho\Vert U(s+\theta'\theta\rho,s)f\Vert_{l'(0)}d\theta'    + C_2\rho\Vert f\Vert_{l'(0)}.
\end{equation*}
Hence from Theorem 7.A and \eqref{7.4} we can see
 \begin{equation} \label{7.10}
 \left\Vert i\frac{U(\ts)f - f}{\rho} - H(t)f \right\Vert \leq  C\rho|f|_{l'}
\end{equation}
for all $f \in \Sspace$ with an integer $l' \geq 0$.  Writing 
 \begin{align} \label{7.11}
 & \Cts(\ts)f - U(\ts)f = \rho\left\{ \frac{\Cts(\ts)f - f}{\rho} - H(t)f \right\} \notag \\
&\quad -  \rho\left\{ \frac{U(\ts)f - f}{\rho} - H(t)f \right\},
\end{align}
we can prove \eqref{7.9} from \eqref{7.5} and \eqref{7.10}.
\end{proof}
Now we will prove Theorem 2.1.  Hereafter we assume $|\Delta| \leq \rho^*$.  We have proved \eqref{2.14} in Corollary 5.9. First we assume $f \in \Sspace$.  From \eqref{5.39} we can write 
\begin{align}  \label{7.12}
   & \kdelta(t,0) f - U(t,0)f  =  {\cal C}(t,\tau_{\nu-1}){\cal C}(\tau_{\nu-1},\tau_{\nu-2})\cdots {\cal C}(\tau_{1},0)f
                          \notag \\
                          &- U(t,\tau_{\nu-1})U(\tau_{\nu-1},\tau_{\nu-2})\cdots U(\tau_{1},0)f 
                          = \sum_{j=1}^{\nu} {\cal C}(t,\tau_{\nu-1}){\cal C}(\tau_{\nu-1},\tau_{\nu-2})\cdots
                          \notag \\
          & 
               \cdot{\cal C}(\tau_{j+1},\tau_j) \bigl\{{\cal C}(\tau_j,\tau_{j-1}) -U(\tau_j,\tau_{j-1})\bigr\}U(\tau_{j-1},0)f.
\end{align}
Using \eqref{5.35}, we have
\begin{equation*}  
    \Vert\kdelta(t,0) f - U(t,0)f \Vert \leq  e^{Kt}\sum_{j=1}^{\nu}\Vert \bigl\{{\cal C}(\tau_j,\tau_{j-1}) -U(\tau_j,\tau_{j-1})\bigr\}U(\tau_{j-1},0)f\Vert,
\end{equation*}
which leads to 
\begin{align}  \label{7.13}
    \Vert\kdelta(t,0) f - U(t,0)f \Vert & \leq  Ce^{Kt}\sum_{j=1}^{\nu}(\tau_j- \tau_{j-1})^{3/2}|U(\tau_{j-1},0)f|_l 
    \notag\\
    & \leq C'\sqrt{|\Delta|}e^{KT}T|f|_{l'}
\end{align}
from \eqref{7.3} and \eqref{7.9}.  Hence we see that as $|\Delta| \to 0$, $\kdelta(t,0)f$ for $f \in \Sspace$ converges to $U(t,0)f$ in $L^2$ uniformly in $t \in [0,T]$.
\par
	Let $f \in L^2$ be arbitrary.  For any $\epsilon > 0$ we take a $g \ \in \Sspace$ such that $\Vert g - f \Vert < \epsilon$.  Then from \eqref{2.14} and \eqref{7.2} we see
\begin{align}  \label{7.14}
   & \Vert K_{\Delta}(t,0)f - U(t,0)f \Vert  \leq \Vert K_{\Delta}(t,0)(f-g)\Vert + \Vert K_{\Delta}(t,0)g - U(t,0)g\Vert  
    \notag \\
          &  + \Vert  U(t,0)(f-g)\Vert \leq \Vert K_{\Delta}(t,0)g - U(t,0)g\Vert +(e^{KT}+1)\Vert f - g \Vert,
          \end{align}
          which shows 
          \begin{equation*}
         \overline{ \lim_{|\Delta| \to 0}}\sup_{0 \leq t \leq T}\Vert K_{\Delta}(t,0)f - U(t,0)f \Vert \leq (e^{KT}+1)\epsilon
          \end{equation*}
          because of $g \in \Sspace$.  This indicates
          \begin{equation*}
          \lim_{|\Delta| \to 0}\sup_{0 \leq t \leq T}\Vert K_{\Delta}(t,0)f - U(t,0)f \Vert = 0.
                    \end{equation*}
	In the end, to complete the proof of Theorem 2.1 we have only to prove \eqref{2.16}.   From \eqref{2.15} and \eqref{5.2} we have
\begin{align*}  
   & S'(t,s;q^{t,s}_{x,y})
        = \frac{m|x - y|^2}{2(t - s)} + \int_{\gamma^{t,s}_{x,y}} \bigl(A'\cdot dx - V'dt\bigr)
        \notag \\
     & =  S(t,s;q^{t,s}_{x,y}) + \psi(t,x) - \psi(s,y),
\end{align*}
which shows \eqref{2.16} from \eqref{5.6} and \eqref{5.7}.
\par
	Next we consider the Lagrangian function defined by \eqref{2.17}.  Let $\mathcal{F}(\theta,s;\qts)\\ (s \leq \theta \leq t)$ be the $l \times l$ matrix defined as the solution to \eqref{2.19}.  The following has been proved in Lemma 3.1 of \cite{Ichinose 2007}.
\begin{lem}
We assume
\begin{equation}  \label{7.15}
    |\partial_x^{\alpha}h_{1jk}(t,x)| \leq C_{\alpha}, \ |\alpha| \geq 1
    \end{equation}
    in $\domain$ for all $j,k = 1,2,\dots,l$.  Then we have
\begin{equation}  \label{7.16}
    |\partial_x^{\alpha}\partial _y^{\beta}\mathcal{F}(t',s;\qts)| \leq C_{\alpha\beta} < \infty
    \end{equation}
    in $0 \leq s  \leq t' \leq t \leq T$ for all $\alpha$ and $\beta$.
\end{lem}
	Using $S(\ts;\qts)$, we define
\begin{equation}  \label{7.17}
 \mathcal{C}_s(t,s)f =  \sqrt{m/(2\pi i\rho)}^{\ d}
                  \int \bigl(\exp iS(t,s; q^{t,s}_{x,y})\bigr)\mathcal{F}(\ts;\qts)f(y)dy
\end{equation}
if $0 \leq s < t \leq T$ and  $\mathcal{C}_s(t,t)f = f$ for $f = {}^t(f_1,\dots,f_l) \in \Sspace(\bR^d)^l$.  Then we have the following correspondingly to Theorem 5.8.
\begin{pro}
We assume \eqref{2.3}, Assumption 2.2 and \eqref{2.22}.  Then $\mathcal{C}_s(t,s)$  on $\Sspace^l$ can be extended to a bounded operator on $(L^2)^l$ and satisfies 
\begin{equation}  \label{7.18}
 \Vert\mathcal{C}_s(t,s)f \Vert \leq e^{K'(t - s)}\Vert f\Vert, \ 0 \leq t - s \leq \rho^*
\end{equation}
for $f = {}^t(f_1,\dots,f_l) \in (L^2)^l$ with a constant $K' \geq 0$, where $\Vert f \Vert^2 = \sum_{j=1}^l\Vert f_j\Vert^2.$
\end{pro}
\begin{proof}
From \eqref{2.19} we have
\begin{equation} \label{7.19}
\mathcal{F} (t',s;q^{t,s}_{x,y})- I =
-i \int_{s}^{t'} H_1(\theta,\qts (\theta))\mathcal{F}
(\theta,s;q^{t,s}_{x,y})d\theta.
\end{equation}
Then from the assumption \eqref{2.22} and Lemma 7.3 we have
\begin{equation} \label{7.20}
\left|\partial_x^{\alpha}\partial_y^{\beta}\left\{ \mathcal{F}
(t,s;q^{t,s}_{x,y})- I \right\}
\right| \leq C_{\alpha,\beta} (t - s)
\end{equation}
for all $\alpha$ and $\beta$.  Using $\Cts(\ts)$ defined by \eqref{5.6}, we can write 
\begin{align} \label{7.21}
& \Cts_s(\ts) f = \Cts(\ts)f + \sqrt{\frac{m}{2\pi i\rho}}^{\,d}   \int \bigl(\exp iS(t,s; 
q^{t,s}_{x,y})\bigr)
  \notag \\
& \quad \times
\left\{\mathcal{F} (t,s;\qts) - I\right\}f(y)dy \equiv \Cts(\ts)f + \Cts'_s(\ts)f.
\end{align}
Then, noting \eqref{7.20}, from Theorems 5.8 and 5.10 we obtain
\begin{align*}
\Vert \Cts_s(\ts) f \Vert &
\leq \mathrm{e}^{K(t - s)}\Vert f \Vert + C_0(t - s)\Vert 
f \Vert 
 \leq \mathrm{e}^{K'(t - s)}\Vert f \Vert
\end{align*}
with constants $C_0 \geq 0$ and $K' \geq 0$, which shows \eqref{7.18}.
\end{proof}
We have the following correspondingly to Proposition 6.6.
\begin{lem}
We consider the equation \eqref{2.18}.  Assume \eqref{2.3}, \eqref{2.8}, \eqref{2.9} and \eqref{7.15}.  Then there exist $r_{jk}(\ts;x,w)\ (j,k = 1,2,\dots,l)$ satisfying \eqref{5.3} with an $M \geq 0$ such that 
\begin{align}  \label{7.22}
     & \left( i\frac{\partial}{\partial t} - H(t)I - H_1(t)\right)\Cts_s(\ts) f
         \notag \\
        & = \sqrt{t - s}\Bigl(R_{sjk}(\ts);j\downarrow k\rightarrow 1,2,\dots,l \Bigr)f
\end{align}
for $f \in \Sspace(\bR^d)^l$.
\end{lem}
\begin{proof}
From \eqref{2.3} and \eqref{2.8} we had \eqref{6.24}.  Hence we can prove \eqref{7.22} from Proposition 3.5 of \cite{Ichinose 2007}.
\end{proof}
	Now we will prove Theorem 2.2.  Let $|\Delta| \leq \rho^*.$  
		Using Proposition 7.4, from \eqref{2.21}, Remark 2.2 and \eqref{7.17} we can write
\begin{equation}  \label{7.23}
     K_{s\Delta}(t,0)f = {\cal C}_s(t,\tau_{\nu-1}){\cal C}_s(\tau_{\nu-1},\tau_{\nu-2})\cdots 
{\cal C}_s(\tau_1,0)f
\end{equation}
for $f \in (L^2)^l$ and get the estimates \eqref{2.14} by \eqref{7.18}. Next  consider the equation \eqref{2.18}.  Then since we are assuming \eqref{2.22}, we get the same assertions as in Theorem 7.A and so  get \eqref{7.3}.  
 Hence we can complete the proof of Theorem 2.2 as in the proof of Theorem 2.1, using Theorem 6.5, Proposition 7.4 and Lemma 7.5.
\section{Proofs of Theorems 2.3 and 2.4}
In this section we always suppose Assumptions 2.3 and 2.4. Let $\mathcal{L}^{\sharp}(t,x,\dot{x})$ be the Lagrangian function defined by \eqref{2.24}.  We write
\begin{equation} \label{8.1}
W(t,x) := 2 \sum_{1 \leq j < k \leq 4}V_{jk}(t,x(j)-x(k)).
\end{equation}
Let's define $\qts$ by \eqref{2.20} and write $\gamma^{t,s}_{x,y}$ as \eqref{5.1}.   Then the classical action for $\qts$ is written as 
\begin{align}  \label{8.2}
   S(t,s;q^{t,s}_{x,y})
       & = \sum_{l=1}^4\Biggl\{\frac{m_l|x(l) - y(l)|^2}{2(t - s)} + \int_{\gamma^{t,s}_{x,y}} \Bigl(A^{(l)}(t,x(l))\cdot dx(l) 
        \notag \\
     & -V_l(t,x(l))dt\Bigr)\Biggr\} - \int_{\gamma^{t,s}_{x,y}}W(t,x)dt
\end{align}
correspondingly to \eqref{5.2}.
\par
  Let $p(x,w)$ be a function satisfying \eqref{5.3} and define $P(t,s)$ by \eqref{5.4}.  Then we have the same assertions as in Lemma 5.1.   We define $\tau(\sigma)$ and $\zeta^{(l)}(\sigma)$ by \eqref{5.12} for $x = x(l)$, $y = y(l)$ and $z = z(l)$.  Hereafter for simplicity we suppose $m = m_l\ (l= 1,2,3,4)$
  \begin{lem}
  Let $p(x,w)$ be a function satisfying \eqref{5.3}.  Then for any $0 < \epsilon \leq 1$ and $0 \leq s < t \leq T$ we have \eqref{5.8}, where
\begin{equation}  \label{8.3}
  \Phi(t,s;x,y,z) = (\Phi^{(1)},\Phi^{(2)},\Phi^{(3)},\Phi^{(4)}) \in \bR^{4d},
\end{equation}
\begin{align}  \label{8.4}
    & \Phi_j\upl =  z(l)_j - \frac{x(l)_j + y(l)_j}{2} + \frac{t - s }{m}
              \int_0^1A_j\upl\bigl(s,x(l) + \theta(y(l) - x(l))\bigr)d\theta           \notag \\
      & -\frac{(t - s)^2 }{m}\int_0^1\int_0^1
                \sigma_1E_j\upl(\tau(\sigma),\zeta\upl(\sigma))
                       d\sigma_1d\sigma_2 \notag \\
                        & - \frac{t - s }{m}\sum_{k=1}^d \bigl(z(l)_k - 
x(l)_k\bigr)\int_0^1\int_0^1\sigma_1B_{jk}\upl
                     (\tau(\sigma),\zeta\upl(\sigma))d\sigma_1d\sigma_2  \notag \\
                 &    + \frac{(t - s)^2 }{m}\int_0^1\int_0^1
                \sigma_1\frac{\partial W}{\partial x(l)_j}(\tau(\sigma),\zeta(\sigma))
                       d\sigma_1d\sigma_2                
\end{align}
or
\begin{align}  \label{8.5}
    &  \Phi_j\upl =  z(l)_j - \frac{x(l)_j + y(l)_j}{2} + \frac{t - s }{m}
              \int_0^1A_j\upl\bigl(s,x(l) + \theta(y(l) - x(l))\bigr)d\theta           \notag \\
      & -\frac{(t - s)^2 }{m}\int_0^1\int_0^1
                \sigma_1E_j\upl(\tau(\sigma),\zeta\upl(\sigma))
                       d\sigma_1d\sigma_2  \notag \\
 & - \frac{(t - s)^2 }{m}\int_0^1d\theta \sum_{k=1}^d \bigl(z(l)_k - x(l)_k\bigr) \int_0^1\int_0^1\sigma_1(1-\sigma_1)\notag \\
                        &\times \frac{\partial B_{jk}\upl}{\partial t}
                     (s+\theta(1-\sigma_1)\rho,\zeta\upl(\sigma))d\sigma_1d\sigma_2 \notag \\
                     &       + \frac{(t - s)^2 }{m}\int_0^1\int_0^1
                \sigma_1\frac{\partial W}{\partial x(l)_j}(\tau(\sigma),\zeta(\sigma))
                       d\sigma_1d\sigma_2.        
\end{align}
  \end{lem}
\begin{proof}
We note 
\begin{align}  \label{8.6}
    & d \Bigl(\sum_{l=1}^4\sum_{j=1}^d A_j\upl(t,x(l))dx(l)_j -\sum_{l=1}^4V_l(t,x(l))dt - W(t,x)dt\Bigr) \notag \\
    & = \sum_{l,j}\left(\frac{\partial }{\partial t}A_j\upl + \frac{\partial }{\partial x(l)_j}V_l + \frac{\partial W}{\partial x(l)_j}\right)dt \wedge dx(l)_j \notag \\
    & + \sum_{l}\sum_{j,k=1}^d \frac{\partial }{\partial x(l)_k}A_j\upl dx(l)_k \wedge dx(l)_j = - \sum_{l,j}E\upl_jdt \wedge dx(l)_j
    \notag \\
    & + \sum_l \sum_{1 \leq j < k \leq d}B_{jk}\upl dx(l)_j \wedge dx(l)_k + \sum_{l,j}\frac{\partial W}{\partial x(l)_j}dt \wedge dx(l)_j.
\end{align}
Then we can prove Lemma 8.1 as in the proof of Lemma 5.2.
\end{proof}
  We have the following consequence from Lemma 5.3.
  \begin{lem}
   We have \eqref{5.14}, where $E = E\upl(t,x(l))\ (l = 1,2,3, 4), C_1 = C_1(l) \geq 0, x = x(l), \zeta(\sigma) = \zeta^{(l)}(\sigma), X = X(l) = (x(l),y(l),z(l)) \in \bR^{3d}$ and $M_* = M_{l*}$.
  \end{lem}
  We set $a = \max\{C_1(l);l= 1,2,3,4\} \geq 0$.
  Let us  write \eqref{5.16} and \eqref{5.17} for $E = E\upl(t,x(l))\ ( l = 1,2,3,4)$ as $E'^{(l)}_0(\ts;x(l),y(l),z(l))$ and $E'^{(l)}_1(\ts;x(l),y(l)\\,z(l))$, respectively.  In the same way we write  \eqref{5.19} and \eqref{5.20} for $B_{jk} = B_{jk}\upl(t,x(l))$ as $B'^{(l)}(\ts;x(l),y(l),z(l))$.  We define
  \begin{equation} \label{8.7}
  E'_{0s}(\ts;x,y,z) = \begin{pmatrix}
  &E'^{(1)}_0 & 0 & 0& 0\\
  & 0& E'^{(2)}_0 &0 & 0 \\
  & 0&0 & E'^{(3)}_0 & 0 \\
 & 0&0 & 0  & E'^{(4)}_0
  \end{pmatrix}.
  \end{equation}
 In the same way we define $E'_{1s}(\ts;x,y,z)$ and $B'_{s}(\ts;x,y,z)$.  Then from \eqref{8.4} and \eqref{8.5} 
 we have
\begin{align} \label{8.8}
& \frac{\partial \Phi}{\partial z}(t,s;x,y,z)  = I + \rho^2  E'_{0s}(t,s;x,y,z) - \frac{\rho^2a}{6} + \rho^2  E'_{1s}(t,s;x,y,z)\notag \\
 & + B'_s(t,s;x,y,z) + \frac{\rho^2}{m}\int_0^1\int_0^1\sigma_1(1-\sigma_1)\frac{\partial^2 W}{\partial x^2}(\tau(\sigma),\zeta(\sigma))d\sigma_1d\sigma_2,
\end{align}
which is correspondent to \eqref{5.18}.
\begin{lem}
There exist constants $\rho^* > 0$ and $\delta > 0$ such that for $0 \leq t-s \leq \rho^*$ and $(x,y,z) \in \bR^{12d}$
we have \eqref{5.21} and 
\begin{equation*} 
\left|\frac{\partial \Phi}{\partial z}(\ts;x,y,z)^{-1}\right| \leq C(1+\rho^2)^{-1},
\end{equation*}
where $1 + \rho^2|X|^{2M_*}$ in \eqref{5.21} is replaced with $\prod_{l=1}^4(1 + \rho^2|X(l)|^{2M_{l*}})$.
\end{lem}
\begin{proof}
Noting the assumption \eqref{2.26}, we can easily prove Lemma 8.3 from Lemma 8.2 and \eqref{8.8} as in the proof of Lemma 5.6.
\end{proof}
We see from Lemma 8.3 that the mapping:$ \bR^{4d}\ni z \to \xi = \Phi(\ts;x,y,z) \in \bR^{4d}$ is homeomorphic if $0 \leq t - s \leq \rho^*.$  We write its inverse mapping as $ \bR^{4d}\ni \xi \to z= z(\ts;x,\xi,y) \in \bR^{4d}$.  Then, noting the assumption \eqref{2.26}, we can prove \eqref{5.32}, \eqref{5.35} and \eqref{5.40} as in the proofs of Proposition 5.7, Theorems 5.8 and 5.10.  In the same way we can prove \eqref{6.19} and \eqref{6.23} as in the proofs of Theorem 6.5 and Proposition 6.6.
\par
 We introduce the weighted Sobolev spaces
\begin{align} \label{8.9}
B'^a(\mathbb{R}^{4d})  := \{f \in & L^2(\mathbb{R}^d);
 \|f\|_a := \|f\| + \notag\\
& \sum_{|\alpha| \leq  2a} \|\partial_x^{\alpha}f\| +
\|\omega_a(\cdot)f\| < \infty\}
 \end{align}
for $a = 1,2, \dots$, where $\omega_a(x) = \sum_{l=1}^4<x(l)>^{2a(M_{l*}+1)}$.
We denote the dual space of $B'^a$ by $B'^{-a}$ and the $L^2$ space by $B'^0$.  Then from Theorem 2.4 in \cite{Ichinose 2018} we get the same assertions as in Theorem 7.A.  Joining the results above, we can prove Theorems 2.3 and 2.4 as in the proofs of Theorems 2.1 and 2.2.

 %
%

\begin{thebibliography}{99} %
%
\bibitem{Albeverio et all} Albeverio, S. A, H{\o}egh-Krohn, R. J., Mazzucchi, S.: 
Mathematical Theory of
Feynman Path Integrals, An Introduction, 2nd corrected and enlarged edition. Lecture Notes in Math. {\bf 523},  Berlin, 
Heidelberg:
Springer-Verlag, 2008
%
%
\bibitem{A-M 2005} Albeverio, S. A, Mazzucchi, S.: Feynman path integrals for polynomially growing potentials. J. Funct. Anal. {\bf 221},
83-121 (2005)
%
%
\bibitem{A-M 2009}  Albeverio, S. A, Mazzucchi, S.: An asymptotic functional-integral solution for the Schr\"odinger equation with polynomial potential. J. Funct. Anal. {\bf 257},
1030-1052 (2009)
%
%
\bibitem{Asada-Fujiwara} Asada, K., Fujiwara, D.:
Structure of  fundamental solutions of the Schr\"odinger equation, Convergence of Feynman's path integrals.  In Japanese, \^{S}ugaku {\bf 33}, 97-119 (1981)
%
%
\bibitem{Berezin-Shubin}
 Berezin,  F. A.,   Shubin, M. A.:
   The Schr\"odinger  Equation. Dordrecht: Kluwer Academic Publishers, 1991
   %
\bibitem{D-K} Daubechies, I., Klauder, J. R.:
Quantum-mechanical path integrals with Wiener measure for all polynomial Hamiltonians. II.
  J. Math. Phys. {\bf 26}, 2239-2256 (1985)
%
\bibitem{Feynman} Feynman, R. P.:
Space-time approach to non-relativistic quantum mechanics. Rev. Mod. Phys. 
{\bf 20}, 367-387
  (1948)
%
\bibitem{Feynman-Hibbs} Feynman, R. P., Hibbs, A. R.:
Quantum Mechanics and Path Integrals. New York: McGraw-Hill, 1965
%
%
\bibitem{Fujiwara 2017} Fujiwara, D.:  Rigorous Time Slicing
Approach to Feynman Path Integrals. Tokyo: Springer Japan, 2017
%
%
\bibitem{Ichinose 1997} 
Ichinose, W.: On the formulation of the Feynman path integral
through broken line paths.   Commun. Math. Phys. {\bf 189}, 17-33 (1997)
%
%
\bibitem{Ichinose 1999} 
Ichinose, W.: On convergence of the Feynman path integral formulated
through broken line paths. Rev. Math. Phys.  {\bf 11}, 1001-1025 (1999)
%
%
%
%
%
%
\bibitem{Ichinose 2007}
Ichinose, W.: A mathematical theory of the  Feynman path integral for the generalized Pauli equations. J. Math. Soc. Japan
{\bf 59}, 649-668 (2007)
%

%
\bibitem{Ichinose 2017} Ichinose, W.: Notes on the Feynman path integral for the Dirac equation.
J. Pseudo-Differ. Oper. Appl.
{\bf 9}, 789-809 (2018)
%
%
\bibitem{Ichinose 2018} 
Ichinose,  W., Aoki, T.: Notes on the Cauchy problem for the self-adjoint and non-self-adjoint  Schr\"odinger equations with polynomially growing potentials. J. Pseudo-Differ. Oper. Appl.
On line First, DOI 10.1007/s11868-019-00301-6 (2019)
%
%
\bibitem{Kumano-go} Kumano-go, H.: Pseudo-differential Operators. 
Cambridge: MIT Press,
1981
%
\bibitem{Mazzucchi} Mazzucchi, S.: Mathematical Feynman Path Integrals and Their Applications. Singapore: World Scientific Publishing Co., 2009
%
%
\bibitem{Mizohata} Mizohata, S.: The Theory of Partial Differential Equations. Cambridge: Cambridge University Press, 1973
%
%
\bibitem{Nelson} Nelson, E.: Feynman integrals and the Schr\"odinger equation.   J. Math. Phys. {\bf 5}, 332-343 (1964)
%
%
\bibitem{Nicola} 
Nicola, F.: Convergence in $L^{p}$ for Feynman path integrals.   Adv. Math. {\bf 294}, 384-409 (2016)
%
%
\bibitem{Peskin-Schroeder} 
Peskin, M. E., Schroeder, D. V.: An Introduction to Quantum Field Theory. Cambridge MA: Westview Press, 1995
%
%
\bibitem{Reed-Simon I} 
Reed, M., Simon, B.:  Methods of Modern Mathematical Physics I: Functional Analysis, Revised and Enlarged Edition. San Diego: Academic Press, 1980
%
%
\bibitem{Reed-Simon II}
Reed, M., Simon, B.: Methods of Modern Mathematical Physics II: Fourier Analysis, Self-adjointness. San Diego: Academic Press, 1975
 %
\bibitem{Schwartz} Schwartz, J. T.: Nonlinear Functional Analysis. New 
York: Gordon and Breach Science Publishers, 1969
%
%
\bibitem{Zworski} Zworski, M.: Semiclassical Analysis. Providence, RI: American Mathematical Society, 2012
%


%
\end{thebibliography}
\end{document}